\documentclass[12pt]{article}
\usepackage{graphicx}
\usepackage[driverfallback=hypertex]{hyperref}
\usepackage{amsmath}
\usepackage{amssymb}
\usepackage{eucal}
\usepackage{mathrsfs}
\usepackage{theorem}
\usepackage{color}

\usepackage{amsmath}
\usepackage{amssymb}
\usepackage{graphicx}
\usepackage{eucal}
\usepackage{mathrsfs}
\usepackage{pifont}
\usepackage{cjhebrew}

\newcommand{\calD}{{\mathcal D}}

\newcommand{\calM}{{\mathcal M}}
\newcommand{\calN}{{\mathcal N}}

\newcommand{\calX}{{\mathcal X}}

\newcommand{\bbC}{{\mathbb C}}

\newcommand{\bbR}{{\mathbb R}}

\newcommand{\bbZ}{{\mathbb Z}}

\newcommand{\brk}[1]{\left(#1\right)}          
\newcommand{\Brk}[1]{\left[#1\right]}          
\newcommand{\BRK}[1]{\left\{#1\right\}}        

\newcommand{\mymat}[1]{\begin{pmatrix} #1 \end{pmatrix}}

\newcommand{\deriv}[2]{\frac{d#1}{d#2}}

\newcommand{\pd}[2]{\frac{\partial#1}{\partial#2}}

\newcommand{\beq}{\begin{equation}}
\newcommand{\eeq}{\end{equation}}
\newcommand{\bsplit}{\begin{split}}
\newcommand{\esplit}{\end{split}}
\newcommand{\baligned}{\begin{aligned}}
\newcommand{\ealigned}{\end{aligned}}

\newcounter{sect}

\providecommand{\R}{\bbR}

\newcommand{\textand}{\quad\text{ and }\quad}
\newcommand{\Textand}{\qquad\text{ and }\qquad}

\newtheorem{theorem}{Theorem}[section]
\newtheorem{lemma}{Lemma}[section]
\newtheorem{proposition}{Proposition}[section]

\newtheorem{definition}{Definition}[section]
\newenvironment{proof}{{\flushleft \emph{Proof}:}}{\hfill\ding{110}}

\newcommand{\g}{\mathfrak{g}}
\newcommand{\euc}{\mathfrak{e}}

\newcommand{\cof}[1]{\vartheta^#1}

\newcommand{\conf}{\varrho}
\newcommand{\M}{{\calM}}
\newcommand{\tM}{\tilde{\M}}

\newcommand{\connEuc}{\nabla^{\euc}}
\newcommand{\connM}{\nabla^{\M}}
\newcommand{\conntM}{\nabla^{\tM}}

\newcommand{\tp}{\tilde{p}}
\newcommand{\tPi}{\tilde{\Pi}}
\newcommand{\tg}{\tilde{\gamma}}
\newcommand{\vp}{\varphi}
\newcommand{\Aff}{\operatorname{Aff}}
\newcommand{\dev}{\operatorname{dev}}

\newcommand{\dist}{\operatorname{dist}}

\newcommand{\End}{\operatorname{End}}
\newcommand{\Deck}{\operatorname{deck}}
\newcommand{\Int}{\operatorname{Int}}
\newcommand{\SO}[1]{\text{SO}(#1)}
\newcommand{\tdev}{\dev_\bbC}
\newcommand{\tconf}{{\conf_\bbC}}

\newcommand{\Comega}[2]{{\omega^#1}_#2}

\newcommand{\COmega}[2]{{\Omega^#1}_#2}

\begin{document}

\title{Metric Description of Defects in Amorphous Materials\thanks{RK was partially funded by the Israel Science Foundation and by the Israel-US Binational Foundation.
JPS  was partially supported by the Israel Science Foundation grant 1321/2009 and the Marie Curie International Reintegration Grant No. 239381.}
}

\author{Raz Kupferman  \\
Institute of Mathematics \\
The Hebrew University \\
 Jerusalem 91904, Israel \\
              \\
                      Michael Moshe  \\
                      Racah Institute of Physics \\ The Hebrew University \\ Jerusalem 91904, Israel \\
\\
        Jake P. Solomon \\
                      Institute of Mathematics \\ The Hebrew University \\Jerusalem 91904, Israel
}

\date{October 2013}

\maketitle

\begin{abstract}
Classical elasticity is concerned with bodies that can be modeled as smooth manifolds endowed with a reference metric that represents local equilibrium distances between neighboring material elements. The elastic energy associated with a configuration of a body in classical elasticity is the sum of local contributions that arise from a discrepancy between the actual metric and the reference metric. In contrast, the modeling of defects in solids has traditionally involved extra structure on the material manifold, notably torsion to quantify the density of dislocations and non-metricity to represent the density of point defects.  We show that all the classical defects can be described within the framework of classical elasticity using tensor fields that only assume a metric structure.  Specifically, bodies with singular defects can be viewed as affine manifolds; both disclinations and dislocations are  captured by the monodromy that maps curves that surround the loci of the defects into affine transformations. Finally,  we show that two dimensional defects with trivial monodromy are purely local in the sense that if we remove from the manifold a compact set that contains the locus of the defect, the punctured manifold can be isometrically embedded in Euclidean space.
\end{abstract}

\section{Introduction}
\subsection{Background}

The study of defects in solids and the mechanical properties of  solids with imperfections is a longstanding theme in material science. There exists a wide range of prototypical crystalline defects, among which are dislocations, disclinations, and point defects. An influential dogma since the 1950s has been to describe defects in solids using differential geometric tools. That is, intermolecular effects are smeared out and the defective solid is modeled as a smooth differentiable manifold equipped with a structure that captures its intrinsic geometry, and notably the defects.

The pioneers in the adoption of geometric fields to describe defects in crystalline solids were Kondo \cite{Kon55},  Bilby et al. \cite{BBS55,BS56}, and  Kr\"oner  \cite{Kro81}. Their approach was motivated by the common practice in crystallography---\emph{Burgers circuits}, that are based on discrete steps with respect to a local crystalline structure, and the identification of defects with a discrepancy between closed loops in real space and closed loops in the discrete crystallographic space.
The theory of Kondo and Bilby connects this crystallographic practice to a general theory of continuously dislocated crystals. Its central pillar is a notion of \emph{parallelism}---it is assumed that one knows how to translate vectors in parallel from one point in the body to another. This premise reflects the existence of an underlying lattice, and parallelism is defined such that the components of the translated vector remain constant with respect to the lattice axes.

For defects of dislocation-type, Bilby et al. assume a property of \emph{distant parallelism}, which amounts to the assumption that a  lattice can be defined globally far enough from the locus of the defect, so that the parallel transport of vectors with respect to this lattice is path-independent.
Mathematically, their assumption amounts to the postulation of a global \emph{frame field}, which automatically defines a connection  with respect to which this frame field is parallel. This connection is the main geometric structure used by Kondo and Bilby et al. to describe the state of a dislocated crystal.

A relation between material properties and a material connection was rigorously defined in Wang \cite{Wan67}.
Wang's point of view is that the mechanical properties of a solid are fully encoded in its constitutive relations. In particular, these relations define a collection of symmetry-preserving maps between tangent spaces at different points. Intuitively, this collection of maps determines in what sense a material element at one point is the ``same" as a material elements at another point.  Mathematically, this collection of maps can be identified with a family of sections of the principal bundle of frame fields. Each such section defines a connection. If the connection is independent of the chosen section,  the material symmetries are said to be associated with a \emph{material connection}. As Wang shows, a material connection exists only  in the case of discrete symmetries.

Every connection carries two associated geometric fields---\emph{curvature} and \emph{torsion}. A non-vanishing curvature implies that vectors transported in parallel along closed loops do not regain their initial value. For a connection that admits a parallel
 frame field, the curvature is always zero, and the manifold is then said to be \emph{flat}. A non-vanishing torsion is associated with an asymmetry in the law of parallel transport, or with the failure of parallelograms to be closed.  Torsion has been traditionally associated with the presence of non-zero Burgers vectors, and more specifically, as a measure for the density of dislocations (see the recent paper of Ozakin and Yavari \cite{OY12} for a modern derivation of the relation between the two). Thus, a dislocated crystal is commonly described as a manifold equipped with a connection that has zero curvature but non-zero torsion, a manifold known as a \emph{Weitzenb\"ock manifold}.

Another thoroughly-studied type of defect is a \emph{disclination}.
Like  dislocations, disclinations are detected by the parallel transport of vectors. If dislocations are translational defects, disclinations are rotational defects. A vector translated in parallel along a closed loop that surrounds a disclination line does not return to its initial value. The geometric interpretation of this effect is the presence of a non-zero curvature. There is in fact a close analogy between dislocations and disclinations: dislocations are measured in terms of a translational Burgers  vector whereas disclinations are measured in terms of an angular \emph{Franck vector}. In a dislocated solid, torsion measures the Burgers vector density, whereas in a disclinated body, curvature measures the density of the Franck vector.
See Romanov, Derezin and Zubov \cite{DZ11} and Yavari and Goriely \cite{YG12c} for recent works on the modeling of disclinations in solids.

A third type of defect is a \emph{point defect}. The two prototypical examples of a point defect are a \emph{vacancy}, in which material is removed and the resulting hole is ``welded", and an \emph{intersticial}, in which extra material is inserted at a point. Point defects have also been described using differential geometric fields, and in this context, Kr\"oner introduced the notion of \emph{non-metricity} \cite{Kro92}---a solid with point defects is modeled as a manifold equipped with a connection and a metric, but the connection is not compatible with the metric, i.e., parallel transport does not preserve the Riemannian inner-product. See Katanaev and Volovick \cite{KV92}, Miri and Rivier \cite{MR02}, and Yavari and Goriely \cite{YG12b} for more recent work.

The description of combinations of the above mentioned defect types has been addressed by several authors. In particular, Katanaev and Volovich \cite{KV92}, consider all three defect types.
Miri and Rivier \cite{MR02} also consider  different types of topological defects, and raise two questions that are especially relevant to the present work: (i) can dislocations and disclinations co-exist? and (ii) do dislocations and disclinations ``survive" in amorphous  solids where there is no lattice to dislocate? Both questions will be addressed below.

Finally, we mention the recent series of papers by Yavari and Goriely \cite{YG12,YG12c,YG12b}, which present a very comprehensive description of the Cartan moving frame formalism, and its uses in the modeling of the three types of defects. In addition to a  systematic use of the Riemann-Cartan formalism to simplify calculations in both Riemannian and non-Riemannian manifolds, these papers make a  distinctive contribution: for each type of defect they solve specific yet representative examples, and calculate residual stresses within the framework of nonlinear elasticity.

\subsection{Outline of results}

The present work is motivated by the recent growing interest in amorphous materials that have a non-trivial intrinsic geometry.
There is  a wealth of recent work on pattern formation in living tissues such as
leaves \cite{LM09}, fungi \cite{DB08}, flowers \cite{FSDM05,MP06}, pines \cite{DVR97},  seed pods  \cite{AESK11}, and isolated cells \cite{AAESK12}. There is also a growing literature on engineered soft materials, e.g., \cite{KES07,WMGTNSK13}.
These tissues can be attributed a natural metric structure, and are essentially isotropic.

The renewed interest in amorphous materials endowed with a non-trivial geometry leads us to reconsider whether there can be defects in a medium that is inherently disordered. This is especially pressing in light of Wang's conclusion that a material connection can be defined unambiguously only in the presence of discrete symmetries. Even before that, we should clarify what makes a solid defect-free.

As our basic model for an isotropic defect-free elastic medium we take a Riemannian manifold
that can be embedded isometrically in three-dimensional Euclidean space. Dislocations, disclinations and point defects are in a sense ``elementary defects", in which the defect-free structure is destroyed in clearly delineated regions (typically lines or points). This observation motivates the definition of  elementary defects in isotropic media.
A body with singular defects is defined as a topological manifold $\calX$ and a closed subspace $\calD$, which we identify as the locus of the defect. Moreover,
$\M=\calX\setminus\calD$ is a endowed with a smooth structure and a Riemannian metric, such that
each point of $\M$ has a neighborhood that is defect-free in the above sense.
Note that this definition is in accordance with the concept of a defect proposed by Volterra more than a century ago \cite{Vol07}.

Our next step is to classify types of singular defects, and in particular, connect this  notion to the classical notions of dislocations, disclinations, and point defects. Moreover, we address the question raised by Miri and Rivier as to whether dislocations and disclinations can co-exist.

The differentiable manifold that consists of the body minus the locus of the defect can be viewed as an \emph{affine manifold} \cite{Ben60}, i.e., a differentiable manifold whose local charts are related to each other by affine transformations. As defects are of topological nature, this affine manifold often has a non-trivial topology. One of the basic tools for studying the topology of a manifold is the \emph{fundamental group}, whose  elements are the homotopy classes of closed loops. In our context, the homotopy class describes how a closed curve goes around the locus of the defect. For affine manifolds there is a well-known map from the fundamental group to affine transformations of Euclidean space. This map, which is a group homomorphism, is known as the \emph{monodromy} \cite{FGH81}.

As we will show, the monodromy encompasses both the notions of disclination and dislocation. The linear part of the affine transformation quantifies disclination, whereas the translational part quantifies dislocation. While this seems to imply, in response to the question raised by Miri and Rivier, that dislocations and disclinations can co-exist, the actual answer is more subtle. The subtlety arises from the necessity to choose a base point in defining the fundamental group. Let us focus on a given homotopy class, and assume the fundamental group is abelian.Then, taking into account the base-point dependence, the linear part of the monodromy can be described by a tensor field and the translation part can be described by a vector field. The tensor field describing the linear part is covariant constant. So, the essential information is contained in the tensor's value at a single point, which in the context of defects, is essentially the classical Franck vector. In contrast, the vector field describing the translation part of the monodromy is only covariant constant if the linear part is trivial. In the context of defects, this means that a dislocation can be quantified by a Burgers vector, i.e., by a quantity that depends on the properties of the defect but not on the point of reference, only in the absence of disclinations.

The structure of this paper is as follows: in Section~\ref{sec:singular_defects} we propose a definition for bodies with singular defects, and show that they can be viewed as affine manifolds. In Section~\ref{sec:monodromy} we describe the monodromy of an affine manifold as a map from its fundamental group to affine transformations of its tangent space. In particular, we relate the monodromy to disclinations and dislocations. In Section~\ref{sec:examples} we analyze in this context the classical defects---disclinations, screw dislocations, edge dislocations, higher-order defects, and point defects. For each example we perform an explicit calculation of the monodromy. Furthermore, we prove a general result for two-dimensional defects asserting that trivial monodromy implies a purely local defect. Finally, our results are summarized and interpreted in Section~\ref{sec:discussion}.


\section{Bodies with singular defects}
\label{sec:singular_defects}

\subsection{Definitions}

Our first step is to define what is a solid body with singular defects in a manner that only relies on its metric, and does not depend on the presence of a lattice structure, whether because there is no such structure or because it does not affect the mechanical properties of the body. We start by defining what  is a body free of defects. We formulate everything in arbitrary dimension $n$;  in most applications $n=3$.

\begin{definition}
A defect-free body is a connected, oriented, $n$-dimensional Riemannian manifold  $(\M,\g)$, possibly with boundary, that is  globally flat, i.e.,  there exists an isometric embedding,
\[
f : \M \to \R^n,
\]
where $\R^n$ is equipped with the standard Euclidean metric, $\euc$; that is, $\g = f^\star\euc$.
\end{definition}

A few words about notations: let $f:\M\to\calN$ be a differentiable mapping between two manifolds. Then $df$ is a linear map  $T\M \to f^* T\calN$, where $f^*T\calN$ is a vector bundle over $\M$, with the fiber  $(f^*T\calN)_p$ identified with the fiber $T_{f(p)}\calN$.
Because of this canonical identification, if $F\to\M$ is a vector bundle and $\Phi: f^*T\calN \to F$, then for $p\in\M$ and $\xi \in T_{f(p)}\calN$, we can unambiguously apply $\Phi$ at $p$ to $\xi$, denoting the result by $\Phi_p(\xi)$.

Let  $\alpha\in\Gamma(\calN;T^*\calN)$ be a co-vector field on $\calN$. We denote by $f^*\alpha\in\Gamma(\M;f^* T^*\calN)$ its pullback  defined by,
\[
(f^*\alpha)_p(v) = \alpha_{f(p)}(v), \qquad p \in \M, \quad v\in T_{f(p)}\calN,
\]
where we applied the above canonical identification on the left hand side.
If we view $\alpha$ as a differential form rather than a section, then its pullback, which is denoted by $f^\star\alpha\in\Gamma(\M;T^*\M)$, is defined by
\[
(f^\star\alpha)_p(v) = \alpha_{f(p)}(df_p(v)), \qquad p \in \M, \quad v\in T_{p}\M,
\]
namely, $f^\star\alpha = (f^*\alpha)\circ df$. In this context, $\g = f^\star\euc$ is defined by
\[
\g(u,v) = (f^*\euc)(df(u),df(v)), \qquad u,v\in T\M.
\]

Having defined a defect-free body, we proceed to define a body with singular defects.

\begin{definition}
\label{def:sing_def}
A body with singular defects is a connected  $n$-dimensional topological manifold $\calX$ and a (not-necessarily connected) closed subspace $\calD\subset \calX$ (called the locus of the defect), such that $\M = \calX\setminus \calD$ is connected and oriented. Moreover, $\M$ is endowed with a smooth structure and a (reference) metric $\g$ that is locally Euclidean: every point $p\in\M$ has an open neighborhood $(U_p,\g|_{U_p})$ that embeds isometrically in $(\R^n,\euc)$. Finally, there do not exist a smooth structure and Riemannian metric on $\calX$ such that the inclusion $\M \hookrightarrow \calX$ is an isometric embedding of Riemannian manifolds.
\end{definition}

The last condition in Definition~\ref{def:sing_def} ensures the presence of a defect. In some cases,  there does exist a metric space structure on $\calX$ such that the inclusion $\M \hookrightarrow \calX$ is an isometric embedding of metric spaces.

\subsection{The affine structure}

A body with singular defects carries a natural \emph{affine structure}. Let $\{(U_\alpha,\vp_\alpha: U_\alpha\to\R^n)\}$ be an atlas of orientation-preserving isometries (which by Definition~\ref{def:sing_def} exists), namely,
\[
\g|_{U_\alpha} = \vp_\alpha^\star \euc.
\]
Let $(U_\alpha,\vp_\alpha)$ and $(U_\beta,\vp_\beta)$ be two local charts whose domains have a non-empty intersection, $V = U_\alpha\cap U_\beta$. Then,
\[
\g|_V = \vp_\alpha^\star \euc|_V = \vp_\beta^\star \euc|_V,
\]
which implies that,
\[
\vp_\beta \circ \vp^{-1}_\alpha: \vp_\alpha(V) \to \vp_\beta(V)
\]
is an isometry between two open sets in $\R^n$, namely a rigid transformation of the form
\[
A(x) = Lx + b,
\]
where $L\in \SO{n}$ and $b\in\R^n$.
A manifold endowed with an atlas such that all coordinate transforms are affine is called an \emph{affine manifold}. Bodies with singular defects form a subclass of affine manifolds in which all local chart transformations are rigid transformations.

Any affine structure comes with a natural connection $\connM$ that is induced from the standard connection $\connEuc$ in Euclidean space. Let $(U_\alpha,\vp_\alpha)$ be a local chart, then for $X,Y\in \Gamma(U_\alpha;TU_\alpha)$ we define,
\[
\connM_X Y = (\vp_\alpha^\star\connEuc)_X Y =
d\vp_\alpha^{-1} \brk{(\vp_\alpha^*\connEuc)_{X} d\vp_\alpha(Y)}.
\]
Note that  $d\vp_\alpha(Y)\in\Gamma(U_\alpha,\vp_\alpha^* T\R^n)$, hence
$(\vp_\alpha^*\connEuc)_{X} d\vp_\alpha(Y) \in \Gamma(U_\alpha,\vp_\alpha^* T\R^n)$,
so that the action of $d\vp_\alpha^{-1}:\vp^*_\alpha T\R^n\to T\calM$ is well-defined.

To show that the definition of the pullback connection is independent of the chosen chart, i.e., coincides on overlapping charts, let $\vp_\beta$ be another chart with the same domain, and let $\vp_{\alpha\beta} = \vp_\alpha\circ\vp_\beta^{-1}$ be the affine coordinate transformation.
Then
\[
\begin{split}
(\vp_\alpha^\star\connEuc)_X Y =
(\vp_\beta^\star \vp_{\alpha\beta}^\star \connEuc)_X Y = \vp_\beta^\star\connEuc,
\end{split}
\]
where in the last step we used the fact that
$\vp_{\alpha\beta}^\star \connEuc = \connEuc$, that is, affine transformations in Euclidean space preserve parallelism.

It follows from the definition of $\connM$ that for every local chart $(U_\alpha,\vp_\alpha)$,
\[
\partial^\alpha_i =
d\vp_\alpha^{-1}\brk{\vp_\alpha^* \pd{}{x^i}}
\]
is a parallel vector field.
To transport a vector in parallel within a given chart $(U_\alpha,\vp_\alpha)$
we write $v\in T_pU_\alpha$ as
\[
v = v^i \partial^\alpha_i|_p.
\]
Denoting by $\Pi_{p}^{q}: T_pU_\alpha \to T_qU_\alpha$  the parallel transport operator, we have
\[
\Pi_{p}^{q} (v) = v^i \partial^\alpha_i|_q.
\]

It is easy to see that the induced connection $\connM$ inherits the properties of the Euclidean connection:

\begin{proposition}
$\connM$ is flat and torsion-free.
\end{proposition}

\begin{proof}
Let $(U,\vp)$ be a local chart and let $X,Y\in\Gamma(U;TU)$.
We have,
\[
\connM_X Y - \connM_Y X = \vp^{-1}_*\brk{\connEuc_{\vp_*(X)} \vp_*(Y) -
\connEuc_{\vp_*(Y)} \vp_*(X)},
\]
where $\vp_{*}(X)$ is the vector field on $\vp(U)$ corresponding to $X.$
Because the Euclidean connection is torsion free,
\[
\connM_X Y - \connM_Y X = \left(\vp^{-1}\right)_*\brk{[\vp_*(X) ,\vp_*(Y)]}= [X,Y],
\]
which proves that the torsion is zero. A similar argument shows that the curvature is zero as well.
\end{proof}

To summarize, every affine manifold is equipped with a natural flat and torsion-free connection. For a locally Euclidean manifold, this connection is the Riemannian Levi-Civita connection.

\section{Monodromy}
\label{sec:monodromy}

The flatness of the connection on $\M$ implies that parallel transport in $T\M$ is path-independent within every homotopy class. If the defective body is simply-connected (as is the case, for example, for point defects in three dimensions), then parallel transport is totally path-independent. Otherwise, parallel transport depends on the homotopy class of the curve. In this case it is useful to analyze the affine manifold through its \emph{universal cover}.

\subsection{The universal cover}

Let $p\in\M$; we denote as usual by $\pi_1(\M,p)$ the \emph{fundamental group} of $\M$ at $p$. Every
$g\in \pi_1(\M,p)$ is a homotopy  class of  curves $\gamma:I\to\M$, satisfying $\gamma(0)=\gamma(1)=p$.
$\pi_1(\M,p)$ is a group with respect to curve concatenation ($gh$ is the homotopy class $[\eta * \gamma]$, where $\gamma\in g$ and $\eta\in h$).
Although the fundamental group depends on the reference point, $\pi_1(\M,p)$ and $\pi_1(\M,q)$ are isomorphic. However, there is no canonical isomorphism unless the fundamental group is abelian (see below).

We next introduce the notion of a universal cover; see e.g., Hatcher \cite{Hat02} for a thorough presentation of the algebraic-topological constructs used in this section.

\begin{definition}
A manifold $\tM$ along with a  map $\pi:\tM\to \M$ is a covering space for $\M$, if for every $p\in\M$ there exists an open neighborhood $U_p\subset\M$, such that $\pi^{-1}(U_p)$ is a union of disjoint sets, each homeomorphic to $U_p$. The map $\pi$ is called a projection, and the set $\pi^{-1}(p)$ is called the fiber of $\tM$ over $p$. A universal cover of $\M$ is a simply-connected covering space.
\end{definition}

It is a known fact that a universal cover  exists if $\calM$ is connected, path-connected, and semi-locally simply-connected, all of which we will assume. Moreover, the universal cover is unique up to isomorphism. The universal cover can be constructed explicitly from the fundamental groupoid of all homotopy classes of curves in $\calM$.

The universal cover $\tM$ inherits the flat torsion-free connection of $\M$, $\conntM = \pi^\star\connM$, namely,
\beq
\conntM_X Y = d\pi^{-1}  \brk{(\pi^*\connM)_{X} d\pi(Y)}
\qquad
X,Y \in\Gamma(T\tM).
\label{eq:conntM}
\eeq
Since $\tM$ is simply-connected and $\conntM$ is flat and torsion-free, it follows that parallel transport in $T\tM$ is path-independent. For $p,q\in\tM$, we  denote by $\tPi_p^q:T_p\tM\to T_q\tM$ the parallel transport operator from $p$ to $q$. For $\gamma : I \to \M$, let
$\Pi_\gamma : T_{\gamma(0)}\M \to T_{\gamma(1)}\M$ denote parallel transport along $\gamma.$ Note that
\beq
\tPi_p^q = d\pi^{-1}_q\circ \Pi_\gamma\circ d\pi_p,
\label{eq:projPi}
\eeq
where $\gamma$ is the projection of a curve that connects $p$ to $q$
An important property of covering spaces is the existence of a \emph{unique lift} of curves (see Proposition 1.34 in \cite{Hat02}):
Let $\gamma:I\to\M$ with $\gamma(0)=p$. For every $\tp\in\pi^{-1}(p)$ there exists a unique curve $\tg:I\to\tM$ satisfying $\pi\circ\tg = \gamma$ (i.e., $\tg$ is a lift of $\gamma$) and $\tg(0) = \tp$.

\subsection{Deck transformations}

Covering spaces come with automorphisms called \emph{deck transformations}, which are homeomorphisms $f:\tM\to\tM$ satisfying
\begin{equation}\label{eq:deck}
\pi\circ f = \pi.
\end{equation}
That is, a deck transformation maps continuously every point in the covering space to a point that lies on the same fiber. It is easy to see that deck transformations form a group under composition, which we denote by $\Deck(\tM)$.

There exists a close connection between $\Deck(\tM)$ and  $\pi_1(\M)$, which we present in the following propositions. The proofs, which are all classical, are given in order to demonstrate how the structure of a manifold is revealed through its universal cover.

\begin{proposition}\label{pr:un}
A deck transformation of a path connected covering space is uniquely determined by its value at a single point.
\end{proposition}

\begin{proof}
Let $\vp_1, \vp_2\in\Deck(\tM)$ satisfy
\[
\vp_1(p) = \vp_2(p) = \tp.
\]
Suppose that there exists a point $q\in\tM$ for which
\[
q_1 = \vp_1(q) \ne \vp_2(q) = q_2.
\]
Let  $\alpha:I\to\tM$ be a curve that connects $q$ with $p$ and let $\beta:I\to\tM$ be a curve  that connects $p$ with $\tp$. The curves
\[
\begin{aligned}
\gamma_1 &= (\vp_1\circ\alpha^{-1})*\beta*\alpha \\
\gamma_2 &= (\vp_2\circ\alpha^{-1})*\beta*\alpha,
\end{aligned}
\]
connect $q$ to $q_1$ and $q_2$ respectively. However,
\[
\pi\circ\gamma_1 = (\pi\circ\alpha^{-1})*(\pi\circ\beta)*(\pi\circ\alpha) = \pi\circ\gamma_2,
\]
and $\gamma_1(0) = \gamma_2(0) = q$, thus violating the unique lift property.  This uniqueness proof shows in fact how to construct the deck transformation given its value at a single point (Figure~\ref{fig:deck1}).

\begin{figure}
\begin{center}
\includegraphics[height=2in]{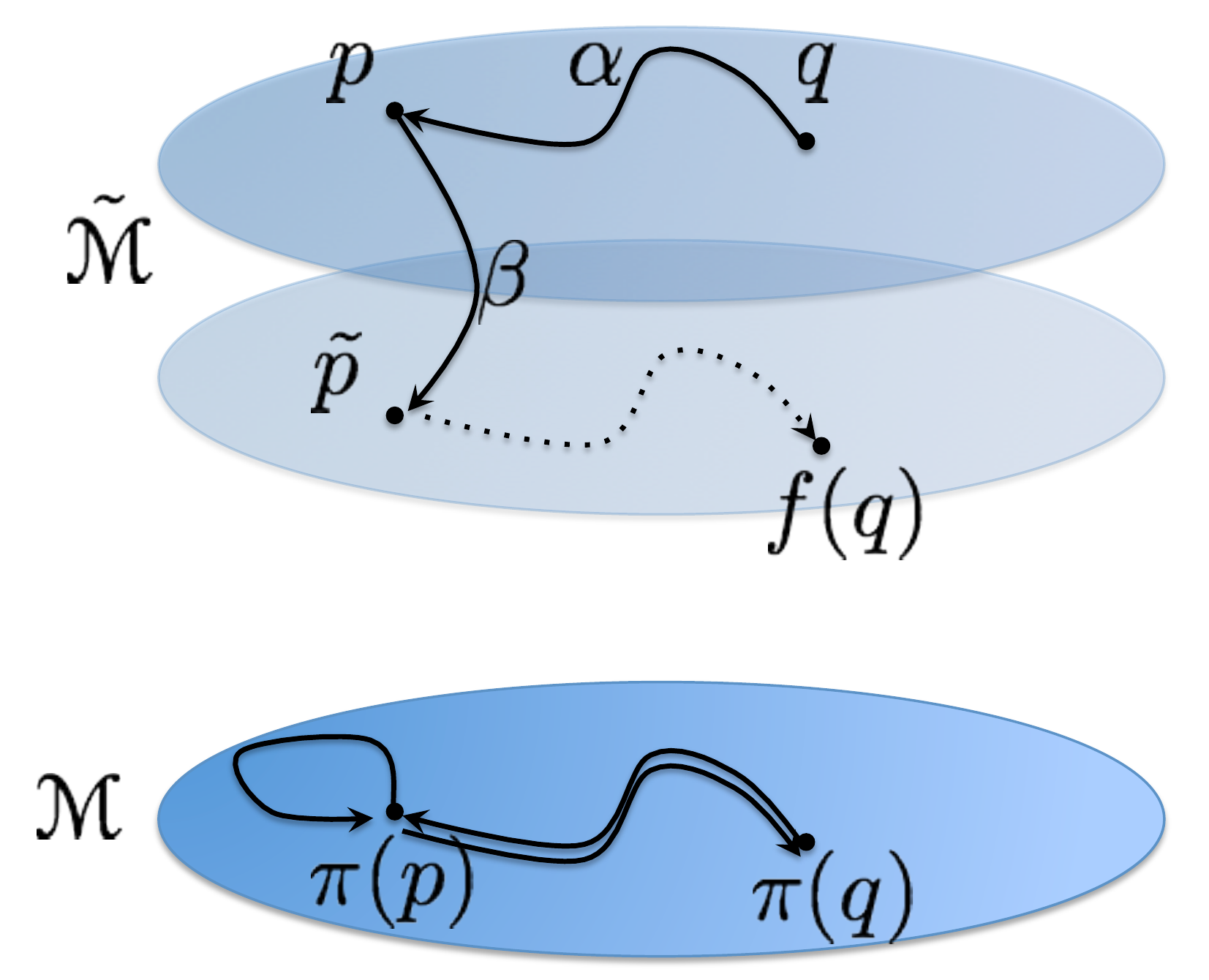}
\end{center}
\caption{Construction of a deck transformation given that $p$ is mapped into $\tp$. }
\label{fig:deck1}
\end{figure}

\end{proof}

\begin{proposition}
\label{prop:g-vpg}
Let $\tM$ be the universal cover of $\M.$ Given a reference point $p\in\tM$, there exists a canonical isomorphism between $\Deck(\tM)$ and  $\pi_1(\M,\pi(p))$. If the fundamental group is abelian, then this isomorphism is independent of the reference point.
\end{proposition}

\begin{proof}
Given a reference point $p$, to every  $\vp\in\Deck(\tM)$ corresponds an element $g\in\pi_1(\M,\pi(p))$ by taking a loop in $\M$ that is the projection of a curve in $\tM$ that connects $p$ to $\vp(p)$. It is easy to see that this mapping is well defined and does not depend on the chosen curve. Conversely, to every $g\in\pi_1(\M,\pi(p))$ corresponds a deck transformation, by mapping $p$ to the end point of the unique lift of a curve $\gamma:I\to\M$ that represents $g$ (see Figure~\ref{fig:deck2}). It is easy to see that this mapping between deck transformations and the fundamental group is both bijective and preserves the group structure, i.e., it is an isomorphism, which we denote by $\Phi(\cdot,p) : \Deck(\tM)\to\pi_1(\M,p)$.

Let $p'$ be another point on the same fiber as $p$. Let $\gamma$ be a curve that connects $p$ to $p'$ and let $\beta$ be a curve that connects $p$ to $\vp(p)$. Then,
\[
\Phi(\vp,p') = \left[\pi \circ \left((\vp\circ\gamma)*\beta*\gamma^{-1}\right)\right] = [\pi \circ \vp\circ\gamma]\cdot[\pi \circ \beta]\cdot[\pi \circ \gamma^{-1}].
\]
Note that
\[
[\pi \circ \beta] = \Phi(\vp,p),
\quad
[\pi \circ \vp\circ\gamma] = [\pi \circ \gamma]
\textand
[\pi \circ \gamma^{-1}] = [\pi \circ \gamma]^{-1}.
\]
Hence
\[
\Phi(\vp,p') =  [\pi \circ \gamma] \cdot \Phi(\vp,p)\cdot  [\pi \circ \gamma]^{-1},
\]
and the mapping $\Phi(\vp,p)$ is independent of $p$ if and only if the fundamental group is abelian.

\begin{figure}
\begin{center}
\includegraphics[height=2in]{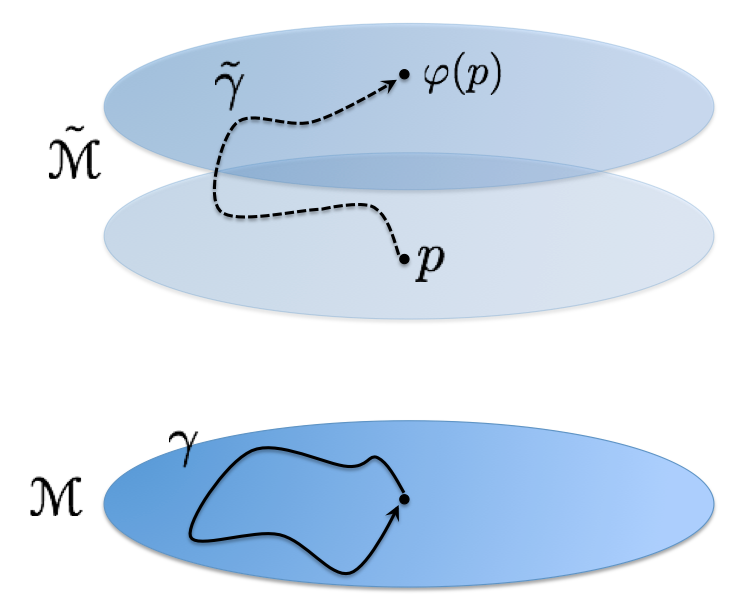}
\end{center}
\caption{The relation between  $\pi_1(\M)$ and $\Deck(\tM)$. Fix a point $p\in\tM$. Given a loop $\gamma$ in $\M$ based at $\pi(p),$ the corresponding deck transformation $\vp$ is the unique one such that $\vp(p)$ is the endpoint of the lift of $\gamma$ starting at $p$. By Proposition~\ref{pr:un}, deck transformations are fully specified by their value at a single point.} 
\label{fig:deck2}
\end{figure}

\end{proof}

Deck transformations leave the connection invariant:

\begin{proposition}
\label{prop:conntM1}
Let $\vp\in\Deck(\tM)$. Then,
\[
\vp^\star\conntM = \conntM.
\]
\end{proposition}

\begin{proof}
Note that both $\conntM$ and $\vp^\star\conntM$ are connections on $T\tM$.
Now,
\[
\vp^\star\conntM = \vp^\star\pi^\star\connM = \pi^\star\connM = \conntM,
\]
where the identity $\vp^\star\pi^\star = \pi^\star$ follows from the identity $\pi\circ\vp = \pi$.
\end{proof}

A consequence of this invariance is that the differential of any deck transformation is a parallel section.

\begin{proposition}
\label{prop:dvp_par}
Let $\vp\in\Deck(\tM)$. Then,
\[
d\vp \in \Gamma(\tM;T^*\tM \otimes \vp^* T\tM)
\]
is a parallel section.
\end{proposition}

\begin{proof}
Let $f:\tM\to\calN$, let $\Phi\in\Gamma(\tM;T^*\tM\otimes f^*T\calN)$, and let $X,v\in\Gamma(\tM;T\tM)$. Denote by $\hat\nabla$ the connection on $T^*\tM \otimes f^* T\calN$ induced by the connections $\conntM$ and $f^*\nabla^\calN.$ By definition,
\[
(f^*\nabla^\calN)_{X} \Phi(v) = (\hat \nabla_X\Phi)(v) + \Phi(\conntM_X v).
\]
In the present case we take $\calN = \tM$, $f = \vp$, and $\Phi = d\vp$.
Hence,
\[
(\hat\nabla_Xd\vp)(v) =  d\vp(\conntM_X v) -  (\vp^*\conntM)_{X} d\vp(v) =
d\vp\brk{\conntM_X v - (\vp^\star\conntM)_X v } = 0,
\]
where in the last step we used Proposition~\ref{prop:conntM1}.
\end{proof}

\subsection{The monodromy of an affine manifold}

The usefulness of the universal cover $\tM$ with its geometric properties pulled back from $\M$ stems from the fact that on the one hand it has patches that are ``genuine replica" of $\M$, and on other hand, parallel transport on its tangent space is path-independent. The ability to transport vectors allows an unambiguous definition of a ``displacement vector" with respect to a reference point, which is known as a \emph{developing map}:

\begin{definition}
Fix $p\in\tM$. The developing map,
\[
\dev:\tM\to T_{p}\tM
\]
is defined as follows: for every $q\in\tM$ choose a curve $\gamma:I\to\tM$ that connects $p$ with $q$. Then,
\[
\dev(q) = \int_0^1 \tPi^{p}_{\gamma(t)}(\dot{\gamma}(t))\, dt.
\]
\end{definition}

Since $\conntM$ is torsion free, this definition is independent of the chosen curve.
This can be shown as follows:
Define a differential 1-form $\omega$ on $\tM$ with coefficients in
$T_p\tM$ by parallel transport, namely,
\[
\omega_q = \tPi_q^p.
\]
The developing map is then given by
\beq
\dev(q) = \int_\gamma \omega,
\label{eq:dev_int}
\eeq
where $\gamma:I\to\tM$ connects $p$ with $q$.
Cartan's first structural equation implies that $d\omega = 0.$ Since $\tM$
is simply-connected, the integral is path independent by
Stokes' theorem.

The developing map is the geometric equivalent of the crystallographic practice of counting discrete steps. Note that the common practice is to define the developing map into $\R^n$ using a trivialization of the tangent bundle; we have reasons that will become apparent below to associated the developing map with a particular tangent space.

Note that both $\tM$ and $T_p\tM$ are affine manifolds.
It follows from \eqref{eq:dev_int} that the developing map preserves the affine structure in the following sense:

\begin{proposition}
\label{prop:ddev}
\[
d\dev_q(v) = \tPi_q^p(v),
\]
where on the right hand side we use the canonical identification of $T_{\dev(q)} T_p\tM$ with $T_p\tM$.
\end{proposition}

The developing map is a useful tool for studying curves in $\M$ that surround defects, by studying how the developing map changes along their lift in $\tM$.
Fix a deck transformation $\vp$. Given $q\in\tM$, we express $\dev(\vp(q))$ by integrating along a curve that connects $p$ and $\vp(q)$ that is a concatenation of the form
\[
\alpha = (\vp\circ\gamma) * \beta,
\]
where $\gamma$ is a curve the connects $p$ and $q$, and $\beta$ connects $p$ with $\vp(p)$ (see Figure~\ref{fig:mono1}).

\begin{figure}
\begin{center}
\includegraphics[height=1.5in]{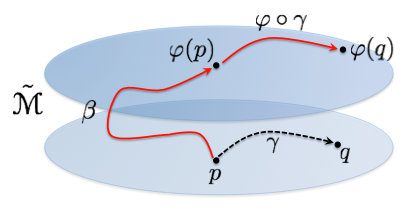}
\end{center}
\caption{The curve $\alpha = (\vp\circ\gamma) * \beta$ (in red) used to calculate $\dev(\vp(q))$.}
\label{fig:mono1}
\end{figure}

By definition,
\beq
\begin{split}
\dev(\vp(q)) &= \int \tPi^{p}_{\alpha(t)}(\dot{\alpha}(t))\,dt \\
&= \int \tPi^{p}_{\beta(t)}(\dot{\beta}(t))\,dt + \int \tPi^{p}_{\vp(\gamma(t))} \circ d\vp_{\gamma(t)} (\dot{\gamma}(t))\,dt  \\
&= \dev(\vp(p)) + \tPi^{p}_{\vp(p)}  \int \tPi^{\vp(p)}_{\vp(\gamma(t))} \circ d\vp_{\gamma(t)}(\dot{\gamma}(t))\,dt,
\end{split}
\label{eq:dev_vpg}
\eeq
where in the passage from the second to the third line we used the composition rule for parallel transport, $\tPi^{p}_{\vp(\gamma(t))} =
 \tPi^{p}_{\vp(p)}  \circ \tPi^{\vp(p)}_{\vp(\gamma(t))}$.
Note that the first term on the right hand side is independent of $q$; it only depends on $\vp$ and on the reference point $p$.

To further simplify the second term on the right hand side we need the following lemma:

\begin{lemma}
\label{lemma:technical}
For every $q\in\tM$,
\beq
\tPi^{\vp(p)}_{\vp(q)} \circ d\vp_{q}  = d\vp_{p} \circ \tPi^{p}_{q}.
\label{eq:useful}
\eeq
\end{lemma}

\begin{proof}
Let $\gamma:I\to\tM$ be a curve that connects $q$ with $p$. Then, the curve $\vp\circ \gamma:I\to\tM$ connects $\vp(q)$ with $\vp(p)$.
Let $\xi(t)$ be a parallel vector field along $\gamma(t)$ and let $\eta(t)$ be a parallel vector field along $\vp\circ\gamma$, satisfying
\beq
\eta(0) = d\vp_{q}(\xi(0))
\label{eq:prop0}
\eeq
(see Figure~\ref{fig:lemma}).
We  show below that
\beq
\eta(t) = d\vp_{\gamma(t)}(\xi(t)).
\label{eq:prop}
\eeq
Since $\eta(t)$ is a parallel vector field
\[
\eta(1) = \tPi_{\vp(q)}^{\vp(p)} \eta(0) = \tPi_{\vp(q)}^{\vp(p)} \circ d\vp_{q}(\xi(0)).
\]
On the other hand, by \eqref{eq:prop} and since $\xi(t)$ is a parallel vector field,
\[
\eta(1) = d\vp_{\gamma(1)}(\xi(1)) = d\vp_{p}\circ  \tPi_{q}^{p}(\xi(0)).
\]
Since this holds for arbitrary $\xi(0)$ we obtain the desired result.

It remains to prove \eqref{eq:prop}. $d\vp_{\gamma(t)}(\xi(t))$ is a vector field along $\vp\circ\gamma$ that satisfies the initial conditions \eqref{eq:prop0}. To show that it is equal to $\eta(t)$ it only remains to show that it is parallel, which follows from the fact that both $d\vp$ (by Proposition~\ref{prop:dvp_par}) and $\xi(t)$ are parallel sections.

\begin{figure}
\begin{center}
\includegraphics[height=1.5in]{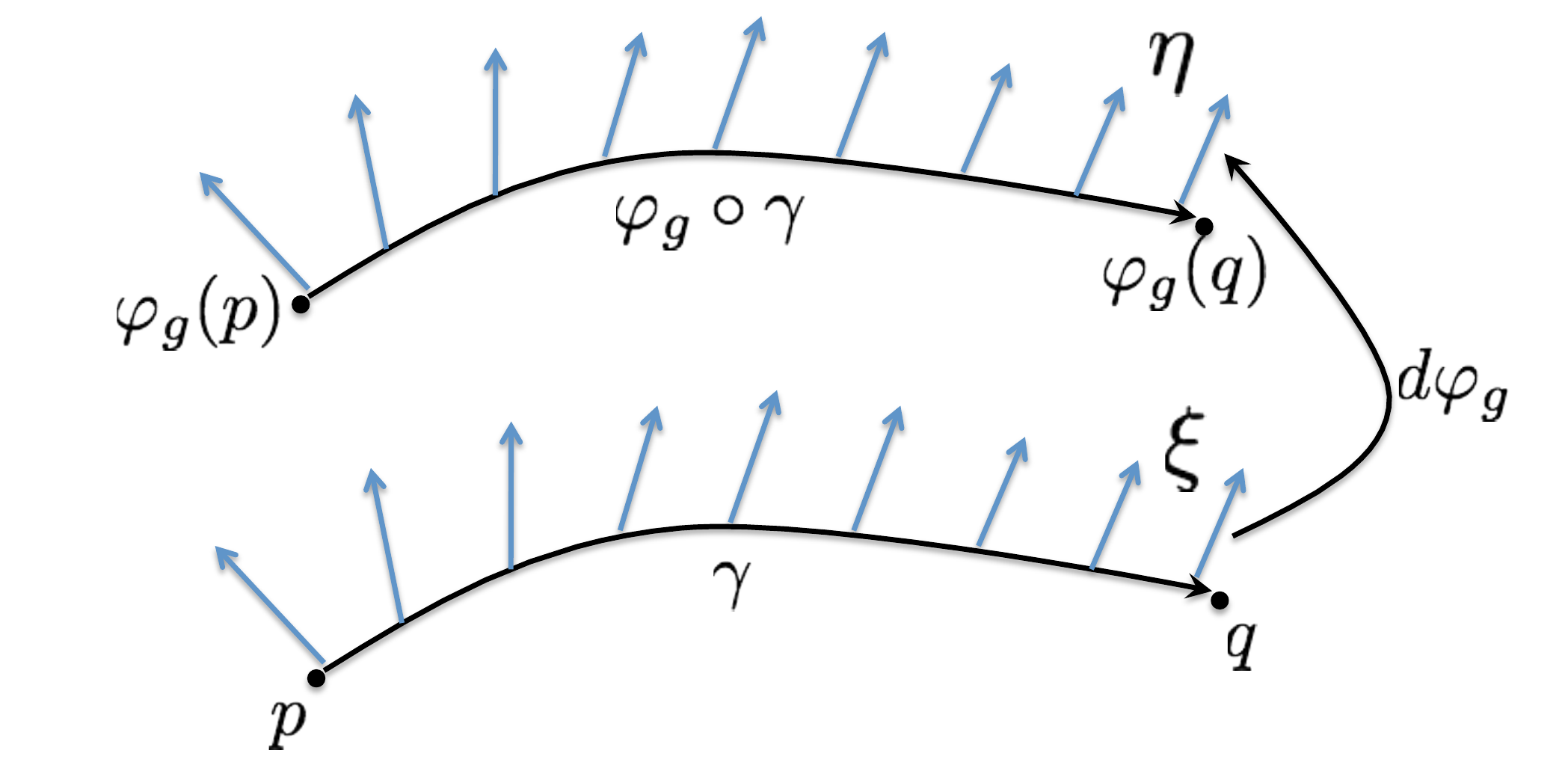}
\end{center}
\caption{Construction for the proof of Lemma~\ref{lemma:technical}.}
\label{fig:lemma}
\end{figure}

\end{proof}

For a vector space $V$, let $\Aff(V)$ denotes the space of affine transformations of $V$.

\begin{theorem}
There exists $m_\vp\in\Aff(T_p\tM)$ such that
\[
\dev(\vp(q)) = m_\vp(\dev(q))
\]
for all $q \in \tM.$ Explicitly, $m_\vp$ is given by
\[
m_\vp (v) = A_\vp\,v + b_\vp,
\]
with $A_\vp\in \End(T_p\tM)$ given by
\beq
A_\vp = \tPi^{p}_{\vp(p)} \circ d\vp_{p},
\label{eq:Ag}
\eeq
and $b_\vp\in T_p\tM$ given by
\beq
b_\vp = \dev(\vp(p)).
\label{eq:bg}
\eeq
\end{theorem}

\begin{proof}
By \eqref{eq:dev_vpg},
\[
\dev(\vp(q)) =  \dev(\vp(p)) + \tPi^{p}_{\vp(p)}  \int \tPi^{\vp(p)}_{\vp(\gamma(t))} \circ d\vp_{\gamma(t)}(\dot{\gamma}(t))\,dt.
\]
By \eqref{eq:useful} with $q = \gamma(t)$,
\[
\tPi^{\vp(p)}_{\vp(\gamma(t))} \circ d\vp_{\gamma(t)}  = d\vp_{p} \circ \tPi^{p}_{\gamma(t)},
\]
hence
\[
\begin{split}
\dev(\vp(q)) &=  \dev(\vp(p))  + \tPi^{p}_{\vp(p)} \circ d\vp_{p}  \int  \tPi^{p}_{\gamma(t)}(\dot{\gamma}(t))\,dt  \\
&=  b_\vp + A_\vp \dev(q).
\end{split}
\]
\end{proof}

Let
\[
m : \Deck(\tM) \longrightarrow \Aff(T_p\tM)
\]
be given by $\vp \mapsto m_\vp.$ As shown next, this mapping is a group homomorphism, and it is called the \emph{monodromy} of the affine manifold $\M$.

\begin{proposition}
$m$ is a group homomorphism, namely,
\[
m_{\vp\circ \psi} = m_\vp\circ m_\psi.
\]
\end{proposition}

\begin{proof}
By definition,
\[
A_{\vp\circ \psi} = \tPi^{p}_{\vp\circ \psi(p)} \circ d\vp_{\psi(p)} \circ d\psi_{p}
= \tPi^{p}_{\vp(p)} \circ \tPi^{\vp(p)}_{\vp\circ \psi(p)} \circ d\vp_{\psi(p)} \circ d\psi_{p}.
\]
We use once again \eqref{eq:useful} to obtain
\[
\tPi^{\vp(p)}_{\vp\circ \psi(p)} \circ d\vp_{\psi(p)} = d\vp_{p} \circ \tPi^{p}_{\psi(p)}.
\]
Hence,
\[
A_{\vp\circ \psi} =\tPi^{p}_{\vp(p)} \circ d\vp_{p} \circ \tPi^{p}_{\psi(p)} \circ d\psi_{p} = A_\vp \circ A_\psi.
\]
Finally, using the fact that $\dev\circ \vp = A_\vp\circ\dev + b_\vp$ and
$\dev(\psi(p)) = b_\psi$,
\begin{align*}
b_{\vp\circ \psi} &= \dev(\vp\circ\psi(p))  \\
&= A_\vp \dev(\psi(p)) + \dev(\vp(p)) \\
&= A_\vp b_\psi + b_\vp.
\end{align*}
Thus,
\[
m_{\vp\circ \psi}(v) = A_{\vp\circ \psi} v + b_{\vp\circ \psi} = A_\vp A_\psi v + A_\vp b_\psi + b_\vp = m_\vp\circ m_\psi(v).
\]
\end{proof}

Having explicit expressions for the affine transformations $m_\vp$, we turn to study their properties. In particular, we allow the reference point to vary and view the developing map as a function
\[
\dev:\tM\times\tM \to T\tM,
\]
such that $\varpi(\dev(p,q)) = p$, where $\varpi:T\tM\to\tM$ is the canonical projection. From this angle, we may view $m$ as a homomorphism from the group of deck transformations to the group of sections of the principal bundle $\Aff(T\tM)$.
The next proposition states that $A_\vp$ is determined by its value at a single point:

\begin{proposition}
$A_\vp$ is a parallel section of $\End(T\tM)$.
\end{proposition}

\begin{proof}
Define $\tPi_\vp \in\Gamma(\tM; \vp^*T\tM\otimes T\tM)$ by
\[
(\tPi_{\vp})_p = \tPi_{\vp(p)}^p.
\]
Then
\[
A_\vp = \tPi_\vp\circ d\vp,
\]
which is a parallel section of $\End(T\tM)$. Indeed,
both $d\vp$ and $\tPi_\vp$ are parallel sections, the first by Proposition~\ref{prop:dvp_par} and the second because parallel transport is path independent.
\end{proof}

The analogous statement for $b_\vp$ is more subtle.

\begin{proposition}
$b_\vp$ is a parallel vector field if and only if $A_\vp$ is the identity section of $\End(T\tM)$.
\end{proposition}

\begin{proof}
By definition,
\[
(b_\vp)_p = \int \tPi^{p}_{\beta_{\vp,p}(t)}(\dot{\beta}_{\vp,p}(t))\,dt.
\]
where $\beta_{\vp,p}$ is a curve that connects $p$ to $\vp(p)$. For $q\in\tM$ set
\[
\beta_{\vp,q} =  (\vp\circ\gamma^{-1}) * \beta_{\vp,p} * \gamma,
\]
where $\gamma$ is a curve from $q$ to $p$. Then,
\[
\begin{split}
(b_\vp)_q &= \tPi_p^q (b_\vp)_p - \int \tPi_{\vp\circ\gamma(t)}^q \circ d\vp_{\gamma(t)}(\dot{\gamma}(t))\, dt +
 \int \tPi^q_{\gamma(t)}(\dot{\gamma}(t))\, dt \\
&= \tPi_p^q (b_\vp)_p  - \int \tPi^q_{\gamma(t)}\circ \tPi_{\vp\circ\gamma(t)}^{\gamma(t)} \circ d\vp_{\gamma(t)}(\dot{\gamma}(t))\, dt +
 \int \tPi^q_{\gamma(t)}(\dot{\gamma}(t))\, dt \\
&= \tPi_p^q (b_\vp)_p  +
 \int \tPi^q_{\gamma(t)} \brk{\text{Id} - (A_\vp)_{\gamma(t)}}(\dot{\gamma}(t))\, dt.
\end{split}
\]
If $A_\vp = \text{Id}$, then
\[
(b_\vp)_q  = \tPi_p^q (b_\vp)_p.
\]
To show that
this condition is also necessary, consider the vector field $(b_\vp)_{\alpha(t)}$ along a curve $\alpha$. Then
\[
\frac{D}{dt} (b_\vp)_{\alpha(t)} = (\text{Id} - (A_\vp)_{\alpha(t)})(\dot{\alpha}(t)).
\]
If at some point along the curve $(A_\vp)_{\alpha(t)} \ne \text{Id}$, then the vector field is not parallel.
\end{proof}

The existence of the monodromy and its properties considered so far only depend on the affine structure of $\M$, and not on its metric properties. The last two properties are of metric nature:

\begin{proposition}
$A_\vp \in O(T\tM)$.
\end{proposition}

\begin{proof}
Since $A_\vp = \tPi_\vp \circ d\vp$, and  both sections are norm preserving, their composition is a norm preserving endomorphism.
\end{proof}

\begin{proposition}
For a fixed reference point,
the developing map $\dev:\tM\to T_p\tM$ is a local isometry.
\end{proposition}

\begin{proof}
This is an immediate consequence of Proposition~\ref{prop:ddev} and the fact that parallel transport is compatible with the metric.
\end{proof}

\subsection{The abelian case: monodromy as sections over $\M$}
\label{sec:abelian}

The universal cover $\tM$ is a useful mathematical construct to study the geometry of the locally-Euclidean manifold $\M$. But from the point of view of material science, one would like to study defects with geometric constructs that are  defined on $\M$ rather than $\tM$. In particular, one would like to make geometric measurements along curves in $\M$ rather than along curves in $\tM$.

The monodromy is a homomorphism $m:\Deck(\tM)\to\Gamma(\tM,\Aff(T\tM))$. Although deck transformations can be related to $\pi_1(\M)$, this relation depends on the choice of a reference point in $\tM$. Hence, in the general case, monodromy cannot be associated with loops in $\M$. The exception is when $\pi_1(\M)$ is abelian, in which case to every element $g\in\pi_1(\M)$ corresponds a deck transformation, which we denote by $\vp_g$. Hence, we consider the monodromy as a homomorphism $m:\pi_1(\M)\to\Gamma(\tM,\Aff(T\tM))$, and for $g\in\pi_1(\M)$ we write $m_g(v) = A_g v + b_g$. Note that we haven't yet got rid of the covering space as the range of $m$ is sections of the principal bundle $\Aff(T\tM)$. However, as will be shown below, $A_g$ can always be identified with a section of $\End(T\M)$, whereas $b_g$ can under specific conditions be identified with a vector field in $\M$.

Before we proceed, we note that the abelian case is quite generic in the following sense: in addition to being directly applicable to both disclinations and screw dislocations, it is also applicable to more complicated defects if one restricts oneself to loops that surround all the defect loci.

\begin{theorem}
Suppose that $\pi_1(\M)$ is abelian.
To every $g\in \pi_1(\M)$
corresponds  a section $A^\M_g \in \Gamma(\M;\End(T\M))$ such that
$A_g = \pi^\star A^\M_g$, namely,
\[
A_g = d\pi^{-1}\circ (\pi^*A^\M_g)\circ d\pi.
\]
Moreover, $(A^\M_g)_p$ is the parallel transport operator along any loop $\gamma$ representing $g$ based at $p$.
\end{theorem}

\begin{proof}
This is an immediate consequence of parallel transport in $T\tM$ being the pullback under $\pi$ of parallel transport in $T\M$. Denote by $\vp_g$ the deck transformation that corresponds to $g$. Such a deck transformation is uniquely determined without reference to basepoint because of the hypothesis that $\pi_1(M)$ is abelian. Let
\[
\tPi_g  : \vp_g^*T\tM \to T\tM
\]
be given by parallel transport.
Then,
\[
A_g = \tPi_g\circ d\vp_g.
\]
Since $\pi$ is a local diffeomorphism, we may write
\begin{align*}
A_g &= d\pi^{-1} \circ \brk{d\pi\circ \tPi_g\circ d\vp_g \circ d\pi^{-1}}\circ d\pi \\
& = d\pi^{-1} \circ \brk{d\pi\circ \tPi_g\circ \vp_g^*(d\pi^{-1})}\circ d\pi,
\end{align*}
where the last transition follows from equation~\eqref{eq:deck}.
It follows from equation~\eqref{eq:projPi} that the expression in parentheses is a section of $\End(\pi^* T\M)$ pulled back from a section of $\End(T\M)$. Namely, evaluated at the point $p \in \M$, it is the pull-back of the parallel transport operator along the loop $\gamma = \pi\circ\beta$, where $\beta$ is a curve that connects $\vp(q)$ to $q$ with $q \in \pi^{-1}(p).$ The choice of $q$ is irrelevant because of the abelian hypothesis.
\end{proof}

\begin{theorem}\label{tm:bg}
Suppose that $\pi_1(\M)$ is abelian.
Then,
\[
b_g = d\pi^{-1}(\pi^* b^\M_g),
\]
where $b^\M_g \in \Gamma(\M;T\M)$ is given by
\[
(b^\M_g)_p = \int \Pi^p_{\gamma,\gamma(t)}(\dot{\gamma}(t))\, dt,
\]
with $[\gamma] = g$.
\end{theorem}

\begin{proof}
Let $\alpha$ be a curve in $\tM$ that connects $p$ and $\vp_g(p)$, and let $\gamma = \pi\circ\alpha$. Then
\[
\begin{split}
b_g &= \int \tPi_{\alpha(t)}^p (\dot{\alpha}(t))\,dt \\
&= d\pi^{-1} \int d\pi\circ \tPi_{\alpha(t)}^p \circ d\pi^{-1} \circ d\pi(\dot{\alpha}(t))\,dt \\
&= d\pi^{-1} \int d\pi\circ \tPi_{\alpha(t)}^p \circ d\pi^{-1} (\dot{\gamma}(t))\,dt \\
&= d\pi^{-1} \pi^* \int \Pi_{\gamma,\gamma(t)}^{\pi(p)} (\dot{\gamma}(t))\,dt,
\end{split}
\]
where in the last passage we used \eqref{eq:projPi}.
\end{proof}

\section{Examples}
\label{sec:examples}

In this section we examine a number of classical  defects using the formalism developed in the previous section.

\subsection{Disclinations}

Disclinations are two-dimensional line defects, i.e., the locus of the defect is a straight line, and the intrinsic geometry of the body is axially symmetric. Isolated disclinations are not common in crystals due to their high energetic cost, but are more common in quasi-two-dimensional systems, such as monolayers of liquid crystals.

Disclinations as topological defects were first introduced by Volterra \cite{Vol07} using the cut-and-weld procedure; see Figure~\ref{fig:disclination}$a$. There are two types of disclinations: positive disclinations, in which a cylindrical wedge is removed and the faces of the cut are welded, and negative disclinations, in which a cylindrical wedge is inserted after a half-plane has been cut (disclinations are often called wedge defects). The sign of the disclination is dictated by the sign of the Gaussian curvature at its locus. In crystals, disclinations are also characterized by either a missing wedge or an extra wedge, in which case the disclination angle is determined by the structure of the unperturbed lattice; see Figure~\ref{fig:disclination}$b$,$c$.

\begin{figure}
\begin{center}
\includegraphics[height=2in]{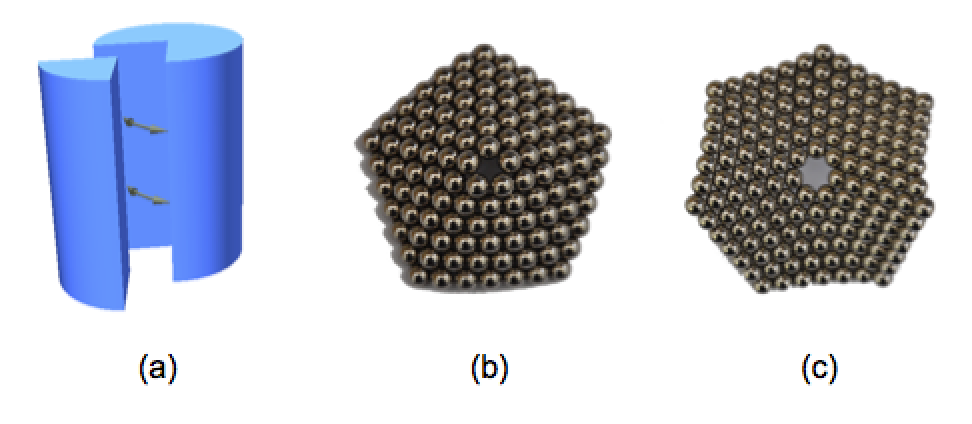}
\end{center}
\caption{(a) Sketch of a positive disclination following Volterra's cut-and-weld procedure.
(b) Positive disclination in an hexagonal lattice; one site has 5 neighbors and its center is a source of positive Gaussian curvature.
(c) Negative disclination in an hexagonal lattice; one site has 7 neighbors and its center is a source of negative Gaussian curvature.
}
\label{fig:disclination}
\end{figure}

Adopting the terminology of the present paper,
disclinations are singular defects, in which the body $\M$ is hemeomorphic to the three-dimensional Euclidean space with a line removed. The fundamental group is isomorphic to the additive group $\pi_1(\M) \cong \bbZ$, where an integer $k\in\bbZ$ corresponds to a loop that surrounds the disclination line $k$ times; the sign of $k$ determines the handedness of the loops. In particular, the fundamental group is abelian so that the results of Section~\ref{sec:abelian} apply.

We parametrize $\M$ using cylindrical coordinates $(R,\Phi,Z)$, with $R>0$, and identifying $\Phi=0$ and $\Phi=2\pi$. The metric on $\M$ can be defined in several equivalent ways:

\begin{enumerate}
\item
{\bfseries Local charts}\,\,
The first way is to construct local charts $\{(U_\beta,\vp_\beta)\}$ and define the metric on $\M$ as the pullback metric, $\g = \vp_\beta^\star\euc$.
Let $\{\Phi_\beta\}$ be a collection of angles. We define local charts
\[
(x_\beta,y_\beta,z_\beta) = \vp_\beta(R,\Phi,Z),
\]
where
\beq
x_\beta= R\,\cos\alpha(\Phi - \Phi_\beta)
\qquad
y_\beta = R\,\sin\alpha(\Phi - \Phi_\beta)
\Textand
z_\beta = Z,
\label{eq:disclination_chart}
\eeq
where $\alpha>0$ ($\alpha>1$ for negative disclinations and  $0<\alpha<1$ for positive disclinations). The range of $\Phi$ has to be smaller than both $2\pi$ and $2\pi/\alpha$.
It can be checked explicitly that all the transition maps $\vp_\beta\circ\vp_\gamma^{-1}$ are rigid rotations.

\item
{\bfseries Orthonormal frame field}\,\,
A second way to define a metric on $\M$ is, following \cite{YG12c}, to introduce a  frame field,
\[
e_1 = \partial_R
\qquad
e_2 = \frac{1}{\alpha R} \partial_\Phi
\qquad
e_3 = \partial_Z,
\]
with dual co-frame,
\[
\cof1 = dR
\qquad
\cof2 = \alpha R\, d\Phi
\qquad
\cof3 = dZ,
\]
and set the metric to be that with respect to which this frame field is  orthonormal, namely,
\[
\g = dR\otimes dR + \alpha^2 R^2\, d\Phi\otimes d\Phi + dZ\otimes dZ.
\]

To show that $(\M,\g)$ is  indeed locally Euclidean we calculate the Riemann curvature tensor using Cartan's formalism.
We first calculate the Levi-Civita connection using \emph{Cartan's first structural equations}.
Introducing an anti-symmetric matrix of connection 2-forms $\Comega\alpha\beta$ satisfying
\[
\nabla_X e_\alpha = \Comega\beta\alpha(X)\, e_\beta,
\]
Cartan's first structural equations are
\[
d\cof\alpha + \Comega\alpha\beta\wedge \cof\beta = 0.
\]
It is easy to check that the only non-zero connection form is
\[
\Comega12 = -\frac{1}{R}\cof2 = -\alpha\,d\Phi.
\]
Note that the fact that the connection form does not vanish implies that the chosen frame field is not parallel with respect to the Riemannian connection.
The curvature form is obtained from \emph{Cartan's second structural equations},
\[
\COmega\alpha\beta = - d\Comega\alpha\beta - \Comega\alpha\gamma\wedge\Comega\gamma\beta.
\]
An explicit substitution shows that the right hand side vanishes, i.e., this Riemannian manifold is locally flat.

\item
{\bfseries Conformal representation}\,\,
A third way of defining a metric  on $\M$ uses the two-dimensional character of disclinations, and the fact that every two-dimensional metric is locally conformal to the Euclidean metric. We adopt a parametrization
\[
\M = \{(X,Y,Z)\in\R^3:\,\, X^2 + Y^2\ne 0 \},
\]
i.e., the $Z$-axis is locus of the disclination. The metric is assumed to be of the following form
\[
\g = e^{2\conf(X,Y)}(dX\otimes dX + dY\otimes dY) + dZ\otimes dZ,
\]
where $\conf(X,Y)$ is the conformal factor.
Liouville's equation states that this metric is locally flat if and only if the Laplacian of $\conf$ vanishes, which implies that flat metrics of this form can be generated by taking $\conf$ to be any harmonic function. A disclination has cylindrical symmetry, and corresponds to
\beq
\conf = \beta\,\log(X^2+Y^2),
\label{eq:conf_disclination}
\eeq
where $\beta$ is the magnitude of the disclination; $\beta>0$ corresponds to the insertion of a wedge, whereas $\beta<0$ corresponds to the removal of a wedge.
\end{enumerate}

We proceed to obtain an explicit expression for parallel transport in $T\M$.
We may use any of the parametrizations introduced above.

\begin{enumerate}
\item
{\bfseries Local charts}\,\,
Take a chart $(U,\vp)$ of the form \eqref{eq:disclination_chart}, with, say,  $\Phi_\beta=0$. Then,
\[
\mymat{\partial_R \\ \partial_\Phi \\ \partial_Z}_{(R,\Phi,Z)} =
\mymat{\cos\alpha\Phi & \sin\alpha\Phi & 0 \\
-\alpha R\,\sin\alpha\Phi & \alpha R\,\cos\alpha\Phi & 0 \\
0 & 0 & 1}
\mymat{\partial_x \\ \partial_y \\ \partial_z}_{(R,\Phi,Z)},
\]
and inversely,
\[
\mymat{\partial_x \\ \partial_y \\ \partial_z}_{(R,\Phi,Z)} =
\mymat{\cos\alpha\Phi & -\frac{1}{\alpha R}\sin\alpha\Phi & 0 \\
\sin\alpha\Phi & \frac{1}{\alpha R}\,\cos\alpha\Phi & 0 \\
0 & 0 & 1}
\mymat{\partial_R \\ \partial_\Phi \\ \partial_Z}_{(R,\Phi,Z)}.
\]
The connection on $T\M$ is the pullback of the Euclidean connection, hence
$(\partial_x,\partial_y,\partial_z)$ is a parallel frame in $TU$. That is,
for $p_0 = (R_0,\Phi_0,Z)$ and $p = (R,\Phi,Z)$,
\[
\Pi_{p}^{p_0} \mymat{\partial_x \\ \partial_y \\ \partial_z}_{p} =
\mymat{\partial_x \\ \partial_y \\ \partial_z}_{p_0}.
\]
Thus, for $v\in T_pU$ of the form
\[
v = v^R\,\partial_R|_p + v^\Phi\,\partial_\Phi|_p + v^Z\,\partial_Z|_p,
\]
we have,
\[
\Pi_p^{p_0}(v) = \mymat{v^R & v^\Phi & v^Z}
\mymat{\cos\alpha(\Phi-\Phi_0) & \frac{1}{\alpha R_0}\sin\alpha(\Phi-\Phi_0) & 0 \\
-\alpha R \sin\alpha(\Phi-\Phi_0) & \frac{R}{R_0}\,\cos\alpha(\Phi-\Phi_0) & 0 \\
0 & 0 & 1}
\mymat{\partial_R \\ \partial_\Phi \\ \partial_Z}_{p_0}.
\]

\item
{\bfseries Conformal representation}\,\,
Since parallel transport along the $Z$ axis is trivial, we focus on the parallel transport of vectors within the $XY$-plane. For that we construct a (local) parallel orthonormal frame. Let $\theta(X,Y)$ be a function to be determined. Any frame field of the form
\beq
\begin{aligned}
e_1 &= e^{-\conf}\brk{\cos\theta\,\partial_X - \sin\theta\,\partial_Y}, \\
e_2 &= e^{-\conf}\brk{\sin\theta\,\partial_X + \cos\theta\,\partial_Y}.
\end{aligned}
\label{eq:par_frame}
\eeq
is orthonormal.
The dual co-frame field is
\[
\begin{aligned}
\cof1 &= e^{\conf}\brk{\cos\theta\,dX - \sin\theta\,dY}, \\
\cof2 &= e^{\conf}\brk{\sin\theta\,dX + \cos\theta\,dY}.
\end{aligned}
\]
The condition for the frame field $(e_1,e_2)$ to be parallel is that $d\cof1 = d\cof2=0$.
It is easy to check that this is satisfied if $\conf(X,Y)$ and $\theta(X,Y)$ satisfiy the Cauchy-Riemann equations,
\beq
\pd{\theta}{X} = -\pd{\conf}{Y} \Textand \pd{\theta}{Y} = \pd{\conf}{X}.
\label{eq:CR}
\eeq
For $\conf$ given by \eqref{eq:conf_disclination},
\[
\theta(X,Y) = 2\beta\,\tan^{-1}\brk{\frac{Y}{X}},
\]
which has a branch cut that can be chosen to be along the negative $X$-axis. Thus, to parallel transport a vector $v\in T_p\M$ to another point on the same plane, we first represent it with respect to the basis $\{e_1,e_2\}$. Since this is a parallel frame, the components of the transported vector remain invariant. We will use this approach when we consider edge dislocations.
\end{enumerate}

We proceed to construct the universal cover $\tM$ of $\M$ in order to put into action the formalism derived in the previous section. Since $\M$ is homotopy equivalent to a circle, whose universal cover is a line, the universal cover of $\M$ is homotopy equivalent to the line. A natural parametrization for $\tM$ is
\[
(R,\Theta,Z) \in (0,\infty)\times\R\times\R,
\]
where the projection map $\pi:\tM\to\M$ is
\[
\pi(R,\Theta,Z) = (R,\Theta \mod 2\pi,Z).
\]

\begin{figure}
\begin{center}
\includegraphics[height=3in]{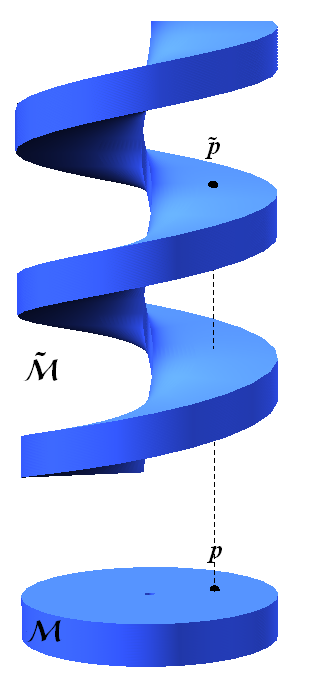}
\end{center}
\caption{The universal cover of the manifold that models a disclination.}
\label{fig:tM1}
\end{figure}

The lift of a loop $\gamma_k\in\pi_1(\M,p)$, $p= (R,\Phi,Z)$, that surrounds the disclination line $k$ times and starts at $(R,\Phi+2\pi\ell,Z)$, ends at $(R,\Phi+2\pi(\ell+k),Z)$. Thus, for $k\in\bbZ$, the corresponding deck transformation is
\[
\vp_k(R,\Theta,Z) = (R,\Theta + 2\pi k,Z)
\]
(the existence of a natural isomorphism from $\pi_1(\M)$ to $\Deck(\tM)$ is due to the former being abelian).

Parallel transport in $T\tM$ is induced by parallel transport in $T\M$.
Let $p = (R,\Theta,Z)$ and $p_0 = (R_0,\Theta_0,Z_0)$.
For a tangent vector $v\in T_p\tM$ of the form
\[
v = v^R\,\partial_R|_p + v^\Theta\,\partial_\Theta|_p + v^Z\,\partial_Z|_p.
\]
we have
\[
\begin{split}
\tPi^{p_0}_{p}(v) &= \mymat{v^R & v^\Theta & v^Z}
\mymat{\cos(\alpha(\Theta-\Theta_0)) & \frac{1}{\alpha R_0}\sin(\alpha(\Theta-\Theta_0)) & 0 \\
-\alpha R\,\sin(\alpha(\Theta-\Theta_0)) & \frac{R}{R_0}\,\cos(\alpha(\Theta-\Theta_0)) & 0 \\
0 & 0 & 1}
\mymat{\partial_R \\ \partial_\Theta \\ \partial_Z}_{p_0}.
\end{split}
\]

With this we proceed to calculate the developing map, $\dev:\tM\to T_p\tM$, for $p = (R_0,\Theta_0,Z_0)$.
Let $\gamma(t) = (R(t),\Theta(t),Z(t))$ be a curve based at $p$.  Then,
\[
\begin{split}
\tPi^{p}_{\gamma(t)} (\dot{\gamma}(t)) &=
\mymat{\dot{R} & \dot{\Theta} & \dot{Z}}
\mymat{\cos(\alpha\Theta) & \sin(\alpha\Theta) & 0 \\
-\alpha R \,\sin(\alpha\Theta) & \alpha R\,\cos(\alpha \Theta) & 0 \\
0 & 0 & 1}
\mymat{\partial_x \\ \partial_y \\ \partial_z}_p \\
&= \deriv{}{t}
\mymat{
R\,\cos(\alpha\Theta)  &
R\,\sin(\alpha\Theta)  &
Z}
\mymat{\partial_x \\ \partial_y \\ \partial_z}_p.
\end{split}
\]
This is easily integrated, yielding
\[
\begin{split}
\dev(R,\Theta,Z) &= \Brk{R\,\cos(\alpha\Theta) - R_0\,\cos(\alpha\Theta_0)}\partial_x|_p \\
&+ \Brk{R\,\sin(\alpha\Theta)  - R_0\,\sin(\alpha\Theta_0)}\partial_y|_p \\
&+ (Z - Z_0)\,\partial_z|_p.
\end{split}
\]
We henceforth represent the developing map as a column vector whose entries are the components of the parallel frame $(\partial_x,\partial_y,\partial_z)$,
\[
\dev(R,\Theta,Z) = \mymat{R\,\cos(\alpha\Theta) - R_0\,\cos(\alpha\Theta_0) \\
R\,\sin(\alpha\Theta)  - R_0\,\sin(\alpha\Theta_0) \\
Z - Z_0}.
\]

From this we easily calculate the monodromy. Let $k\in\bbZ$. Then for $q=(R,\Theta,Z)$,
\[
\dev(\vp_k(q)) = \dev(R,\Theta + 2\pi k,Z) =
A_k \,\dev(q) + b_k,
\]
where
\[
A_k =
\mymat{
\cos 2\pi\alpha k  & -\sin 2\pi\alpha k  & 0 \\
\sin 2\pi\alpha k & \cos 2\pi\alpha k   & 0\\
0 & 0 & 1},
\]
and
\[
b_k = \brk{A_k - I} \mymat{R_0\,\cos(\alpha\Theta_0) \\ R_0\,\sin(\alpha\Theta_0) \\ 0}.
\]
As expected, the linear part $A_k$ of the monodromy is a parallel section of $\End(T\tM)$ (its representation in a parallel frame does not depend on the coordinates). It is a rotation by an angle $2\pi\alpha k$ about the $Z$-axis. Unless $A_k-I=0$ the translation part $b_k$ is not a parallel vector field.

\subsection{Screw dislocations}

Screw dislocations are also line defects, and like disclinations, their intrinsic geometry is axially symmetric.  Like disclinations, screw dislocations were first introduced by Volterra (albeit the term was only coined later). The Volterra procedure for creating a screw dislocation is to  cut a half-plane in the body and weld it with a fixed  offset parallel to the half-plane's boundary; see Figure~\ref{fig:screw}$a$. A visualization of a screw dislocation in a lattice is shown in  Figure~\ref{fig:screw}$b$.

\begin{figure}
\begin{center}
\includegraphics[height=2.4in]{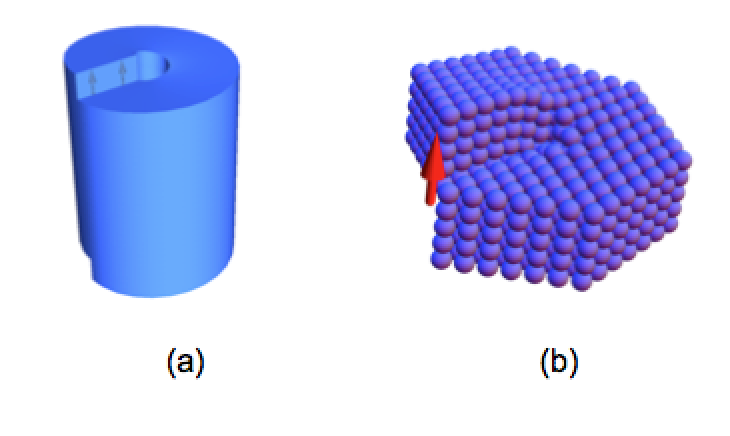}
\end{center}
\caption{(a) Sketch of a screw dislocation following Volterra's cut-and-weld procedure.
(b) Screw disclocation in a lattice structure. }
\label{fig:screw}
\end{figure}

The topology of a body with a screw dislocation is identical to that of a body with a disclination; it is only the metric that differs. We parametrize $\M$ using the same cylindrical coordinates $(R,\Phi,Z)$ as for disclinations.

Once again, there are several alternatives for prescribing the metric on $\M$:

\begin{enumerate}
\item
{\bfseries Local charts}\,\,
Let
\[
x = R\,\cos(\Phi - \Phi_\beta) \qquad
y = R\,\sin(\Phi - \Phi_\beta)
\Textand
z = Z - h(\Phi - \Phi_\beta),
\]
where $h$ is a fixed offset and $\Phi\in(0,2\pi)$. The metric on $\M$ is then the pullback metric.

\item
{\bfseries Orthonormal frame field}\,\,
Following \cite{YG12}, we define an orthonormal  frame field,
\[
e_1 = \partial_R \qquad
e_2 = \frac{1}{R}\partial_\Phi + \frac{h}{R} \partial_Z \qquad
e_3 = \partial_Z.
\]
The dual co-frame is
\[
\cof1 = dR \qquad \cof2 = R\,d\Phi \qquad \cof3 = dZ - h\,d\Phi.
\]
The metric with respect to which this frame field is orthonormal is:
\[
\g = dR\otimes dR + (R^2 + h^2)\, d\Phi\otimes d\Phi + dZ \otimes dZ
- h\,(d\Phi\otimes dZ + dZ\otimes d\Phi).
\]
To show that $(\M,g)$ is locally Euclidean we use again Cartan's formalism.  The only non-zero connection form is
\[
\Comega12 = - d\Phi,
\]
so $\COmega\alpha\beta=0$.
\end{enumerate}

To obtain an explicit formula for the parallel transport of vectors, we follow a procedure similar to that for disclinations. For $p=(R,\Phi,Z)$,
\[
\mymat{\partial_R \\ \partial_\Phi \\ \partial_Z}_p =
\mymat{\cos\Phi & \,\sin\Phi & 0 \\
-R\,\sin\Phi &  R\,\cos\Phi & -h \\
0 & 0 & 1}
\mymat{\partial_x \\ \partial_y \\ \partial_z}_p,
\]
and conversely,
\[
\mymat{\partial_x \\ \partial_y \\ \partial_z}_p =
\mymat{\cos\Phi & -\frac{1}{R} \sin\Phi & -\frac{h}{R} \sin\Phi \\
\sin\Phi &  \frac{1}{R}\cos\Phi & \frac{h}{R} \cos\Phi \\
0 & 0 & 1}
\mymat{\partial_R \\ \partial_\Phi \\ \partial_Z}_p.
\]
Since $(\partial_x,\partial_y,\partial_z)$ is a parallel frame, it follows that
for $v\in T_pU$,  of the form
\[
v = v^R\,\partial_R|_p + v^\Phi\,\partial_\Phi|_p + v^Z\,\partial_Z|_p,
\]
we have
\[
\Pi_p^{p_0}(v) = \mymat{v^R & v^\Phi & v^Z}
\mymat{\cos(\Phi-\Phi_0) & \frac{1}{R_0} \sin(\Phi-\Phi_0) & \frac{h}{R_0} \sin(\Phi-\Phi_0) \\
-R \sin(\Phi-\Phi_0) &  \frac{R}{R_0}\cos(\Phi-\Phi_0) & -h + \frac{hR}{R_0} \cos(\Phi-\Phi_0) \\
0 & 0 & 1}
\mymat{\partial_R \\ \partial_\Phi \\ \partial_Z}_{p_0}.
\]

The universal cover of $\M$ is constructed identically  to disclinations, taking the open half-space, along with the projection
\[
\pi(R,\Theta,Z) = (R,\Theta\mod2\pi,Z).
\]
For $p=(R,\Theta,Z)$, $p_0 =(R_0,\Theta_0,Z_0)$, and a tangent vector $v\in T_p\tM$ of the form
\[
v = v^R\,\partial_R|_p + v^\Theta\,\partial_\Theta|_p + v^Z\,\partial_Z|_p,
\]
we have,
\[
\begin{split}
\tPi^{p_0}_{p}(v)  = \mymat{v^R & v^\Theta & v^Z}
\mymat{\cos(\Theta-\Theta_0) & \frac{1}{R_0} \sin(\Theta-\Theta_0) & \frac{h}{R_0} \sin(\Theta-\Theta_0)  \\
- R\,\sin(\Theta-\Theta_0) &  \frac{R}{R_0}\,\cos(\Theta-\Theta_0) &   \frac{hR}{R_0}\,\cos(\Theta-\Theta_0) - h \\
0 & 0 & 1}
\mymat{\partial_R \\ \partial_\Theta \\ \partial_Z}_{p_0}.
\end{split}
\]

We then calculate the developing map, $\dev:\tM\to T_p\tM$, $p = (R_0,\Theta_0,Z_0)$.
Let $\gamma(t) = (R(t),\Theta(t),Z(t))$ be a curve based at $p$.  Then,
\[
\begin{split}
\tPi^{p}_{\gamma(t)}(\dot{\gamma}(t)) &=
\mymat{
\dot{R} \cos\Theta - \dot{\Theta} R\,\sin\Theta &
\dot{R} \sin\Theta + \dot{\Theta} R\, \cos\Theta &
\dot{Z} - h \dot{\Theta}}
\mymat{\partial_x \\ \partial_y \\ \partial_z}_p
\\
&=
\deriv{}{t} \mymat{R\,\cos\Theta & R\,\sin\Theta & Z - h\Theta}
\mymat{\partial_x \\ \partial_y \\ \partial_z}_p.
\end{split}
\]
Writing the developing map as a column vector whose entries are the components of the parallel frame $(\partial_x,\partial_y,\partial_z)$,
\[
\dev(R,\Theta,Z) = \mymat{R\,\cos\Theta-R_0\,\cos\Theta_0 \\ R\,\sin\Theta - R_0\,\sin\Theta_0 \\ (Z - Z_0) - h(\Theta - \Theta_0)}.
\]

We turn to calculate the monodromy. For $k\in\bbZ$ and $q=(R,\Theta,Z)$,
\[
\dev(\vp_k(q)) = A_k\,\dev(q) + b_k,
\]
where
\[
A_k = I
\Textand
b_k = \mymat{0 \\ 0 \\ -2\pi h k}.
\]
Thus, for a screw dislocation the linear part of the monodromy is the identity section of $\End(T\tM)$, and therefore, as expected, the components of the translational part are independent of $q$, namely, the translational part is a parallel vector field. Moreover, it is a vector field parallel to the $z$-axis, i.e., parallel to the locus of the dislocation, as is expected for the Burgers vector of a screw dislocation.

\subsection{Edge dislocations}

Edge dislocations, like disclinations, are planar defects, that is the geometry of the body is axially symmetric.
The Volterra cut-and-weld procedure that generates an edge dislocation is depicted in Figure~\ref{fig:edge}$a$.
Like in a screw dislocation, the body is cut by a half-plane, however  it is welded with a fixed offset perpendicular to the half-plane's boundary. If $b$ denotes the fixed offset, then the locus of the dislocation has to be a slit whose length is $b$. A visualization of an edge dislocation in a lattice is shown in Figure~\ref{fig:edge}$b$. It is created by an extra half-plane of atoms inserted through the lattice, distorting nearby planes of atoms. Note that it is not a priori clear why both visualizations correspond to the same type of defect.

\begin{figure}
\begin{center}
\includegraphics[height=2.4in]{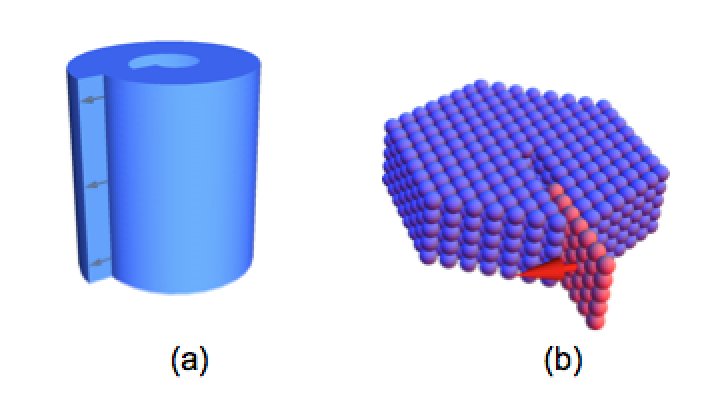}
\end{center}
\caption{(a) Sketch of an edge dislocation following Volterra's cut-and-weld procedure. (b) Edge dislocation in a lattice structure.}
\label{fig:edge}
\end{figure}

Another description of an edge dislocation is as a pair of wedge disclinations of opposite magnitudes. Thus, if a disclination is viewed as a two-dimensional point charge of Gaussian curvature, an edge dislocation should be viewed as a dipole of Gaussian curvature. In an hexagonal lattice, the common occurrence of a wedge-anti-wedge pair is in the form of a pentagon-heptagon pair \cite{SN88}; see Figure~\ref{fig:5-7}.

\begin{figure}
\begin{center}
\includegraphics[height=2in]{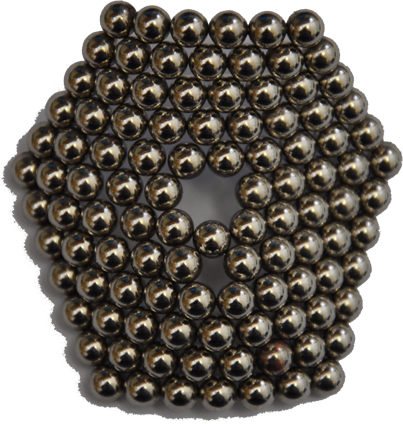}
\end{center}
\caption{A pentagon-heptagon pair in an hexagonal lattice.}
\label{fig:5-7}
\end{figure}

Since the geometry of this defect is two-dimensional we will limit our analysis to a plane, which we parametrize using the coordinates $(X,Y)$.
We define the metric of an edge dislocation using a two-dimensional conformal representation,
\[
\g = e^{2\conf(X,Y)}\brk{dX\otimes dX + dY\otimes dY}.
\]
Recall that such a metric is locally-Euclidean if and only if $\conf$ is harmonic. Since we model an edge dislocation as a wedge anti-wedge pair, the locus of the defect is a pair of parallel lines (points in two dimensions), which we take to be
\[
\{p_1=(-a,0), p_2=(a,0)\}.
\]
A conformal factor that corresponds to two disclinations of opposite signs is,
\beq
\conf(X,Y) = \beta\BRK{\log[(X-a)^2 + Y^2] - \log [(X+a)^2 + Y^2]}.
\label{eq:conf_edge}
\eeq
Note that the coordinates $(X,Y)$ are not Euclidean coordinates, hence the distance between the two defects lines is not $2a$, but rather
\[
\int_{-a}^a e^{\conf(X,0)} \, dX =
a \int_{-1}^1 \brk{\frac{X-1}{X+1}}^{2\beta}\, dX.
\]
This distance is finite for $-1/2<\beta<1/2$.

To obtain explicit expressions for  parallel transport  in $T\M$ we solve the Cauchy-Riemann equations \eqref{eq:CR} for $\conf$ given by \eqref{eq:conf_edge}, yielding
\[
\theta(X,Y) = 2\beta \BRK{\tan^{-1}\brk{\frac{Y}{X-a}} - \tan^{-1}\brk{\frac{Y}{X+a}} },
\]
which can be defined as a smooth function with a branch cut on the segment along the $X$ axis that connects the two loci of the defect. $\theta$ has a jump discontinuity of magnitude $4\pi\beta$ across the branch cut (See Figure~\ref{fig:conf_theta}).

\begin{figure}
\begin{center}
\includegraphics[height=1.5in]{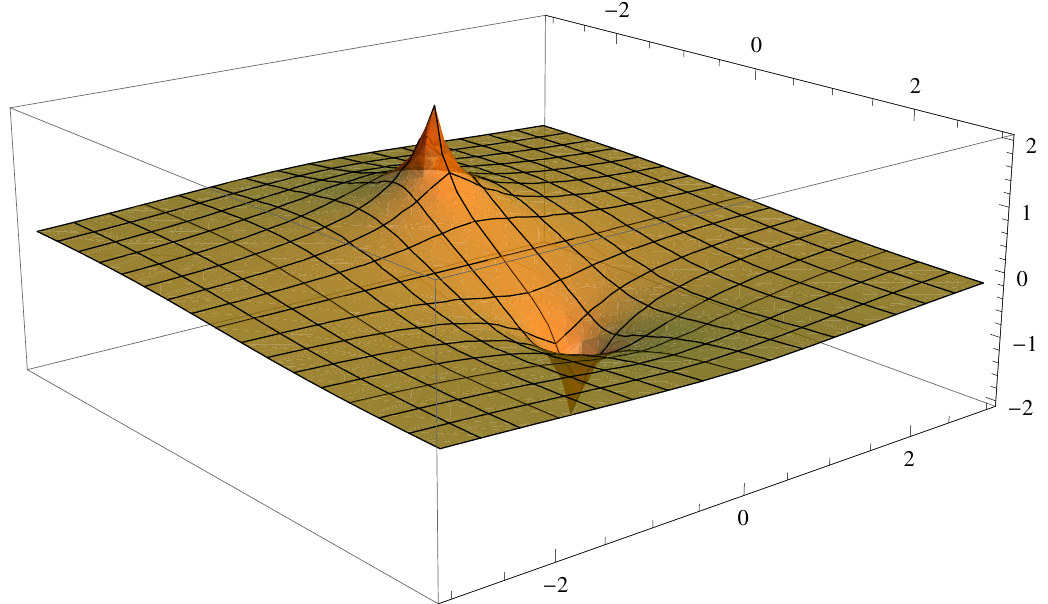}
\hspace{1cm}
\includegraphics[height=1.5in]{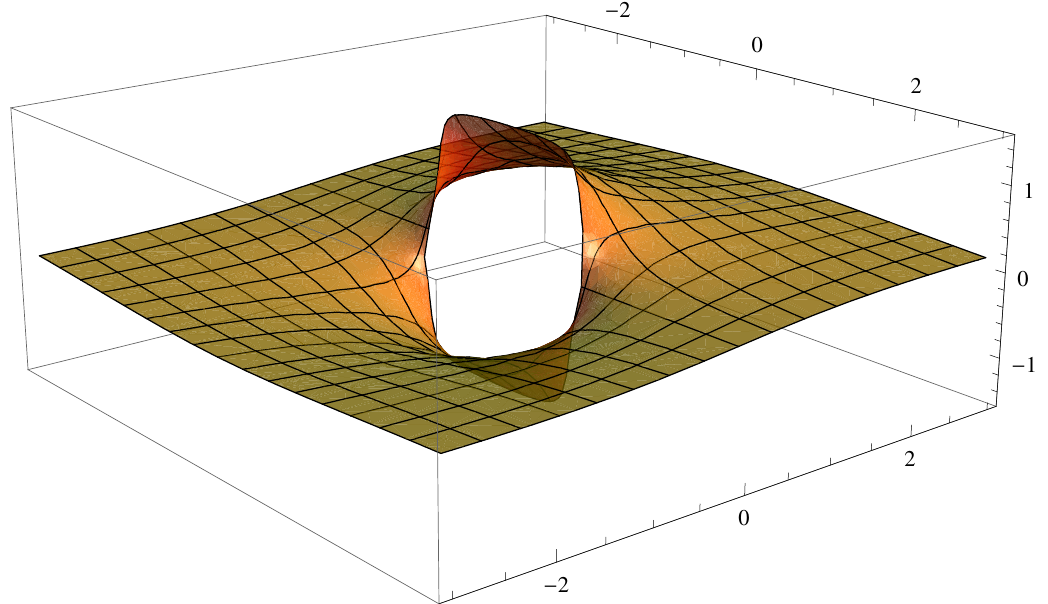}
\end{center}
\caption{$\conf(X,Y)$ and $\theta(X,Y)$ for $\beta=1/4$ and $a=1$.}
\label{fig:conf_theta}
\end{figure}

With the aid of $\theta(X,Y)$, we can prescribe how to parallel transport vectors as long as they do not cross the branch cut.
The frame field $\{e_1,e_2\}$ given by \eqref{eq:par_frame} is parallel. Note that the fact that there exists a global parallel frame field at the exterior of a bounded domain indicates the existence of  distant parallelism, which is a characteristic of edge dislocations.

Inverting \eqref{eq:par_frame} we get
\[
\begin{aligned}
\partial_X &= e^{\conf}(\cos\theta\,e_1 + \sin\theta\,e_2) \\
\partial_Y &= e^{\conf}(-\sin\theta\,e_1 + \cos\theta\,e_2).
\end{aligned}
\]
Thus, for $p=(X,Y)$ and a tangent vector $v\in T_p\M$ written as
\[
\begin{split}
v &= v^X \,\partial_X|_p + v^Y\,\partial_Y|_p \\
&= e^{\conf(p)} \brk{v^X\cos\theta(p) -v^Y\sin\theta(p)}\,e_1|_p + e^{\conf(p)} \brk{v^X\sin\theta(p) + v^Y\cos\theta(p)}\,e_2|_p,
\end{split}
\]
the parallel transport of $v$ to $p_0 = (X_0,Y_0)$ via a curve that does not pass between the two loci of the defect yields,
\[
\begin{split}
\Pi_p^{p_0}(v) &= e^{\conf(p)} \brk{v^X\cos\theta(p) -v^Y\sin\theta(p)}\,e_1|_{p_0} + e^{\conf(p)} \brk{v^X\sin\theta(p) + v^Y\cos\theta(p)}\,e_2|_{p_0} \\
&= e^{\conf(p) - \conf(p_0)} \brk{v^X\cos\theta(p) -v^Y\sin\theta(p)}\brk{\cos\theta(p_0) \partial_X|_{p_0} - \sin\theta(p_0)\partial_Y|_{p_0}} \\
&\,\,\,+ e^{\conf(p) - \conf(p_0)} \brk{v^X\sin\theta(p) + v^Y\cos\theta(p)} \brk{\sin\theta(p_0) \partial_X|_{p_0}  + \cos\theta(p_0)\partial_Y|_{p_0}}.
\end{split}
\]
In particular, for later use
\begin{equation}\label{eq:ptx}
\begin{split}
\Pi_p^{p_0}(\partial_X|_p) &= e^{\conf(p) - \conf(p_0)} \brk{ \cos(\theta(p) - \theta(p_0)) \partial_X|_{p_0} +  \sin(\theta(p) - \theta(p_0))\partial_Y|_{p_0}} .
\end{split}
\end{equation}

We now turn to calculate the monodromy of an edge dislocation. The fundamental group of a plane with two punctures is the free group on two generators. Namely, the first generator is the homotopy class of loops that circle once around the point $(-a,0)$ but do not encircle the other point. The other generator is the homotopy class of loops that circle once around the point $(a,0)$. Indeed, the doubly punctured plane is homotopic to a figure-of-eight. Since we are interested in the structure induced by both defect lines together rather than each separately, we redefine the defect locus $\calD$ to be the closed segment $[-a,a]\times\{0\}$. In particular, since the fundamental group is abelian, we can associate the monodromy with sections over $\M$.

The existence of a ``distant" parallel frame field implies that the linear part of the monodromy is trivial, namely
\[
A_{g} = \text{Id}.
\]
To calculate the translational part of the monodromy, we take an arbitrary reference point $p_0$, denote $\conf_0 = \conf(p_0)$, $\theta_0 = \theta(p_0)$, and take the contour shown in Figure~\ref{fig:contour}. This contour represents the generator $g$ of the fundamental group. Then, using equation~\eqref{eq:ptx}, we obtain
\[
\begin{split}
b_{g}|_{p_0} &= e^{-\conf_0} \dist(p_1,p_2) \Brk{\cos(2\pi\beta-\theta_0) - \cos(2\pi\beta + \theta_0)} \partial_X|_{p_0} \\
&+ e^{-\conf_0} \dist(p_1,p_2) \Brk{\sin(2\pi\beta-\theta_0) + \sin(2\pi\beta + \theta_0)} \partial_Y|_{p_0} \\
&= 2 e^{-\conf_0} \dist(p_1,p_2) \sin (2\pi\beta) \brk{\sin \theta_0 \,\partial_X|_{p_0}  +\cos\theta_0  \,\partial_Y|_{p_0}}.
\end{split}
\]
For $|p_0|\gg1$ we get,
\[
b_{g}|_{p_0} \approx 2 \dist(p_1,p_2) \sin (2\pi\beta) \, \partial_Y|_{p_0}.
\]

\begin{figure}
\begin{center}
\includegraphics[height=2in]{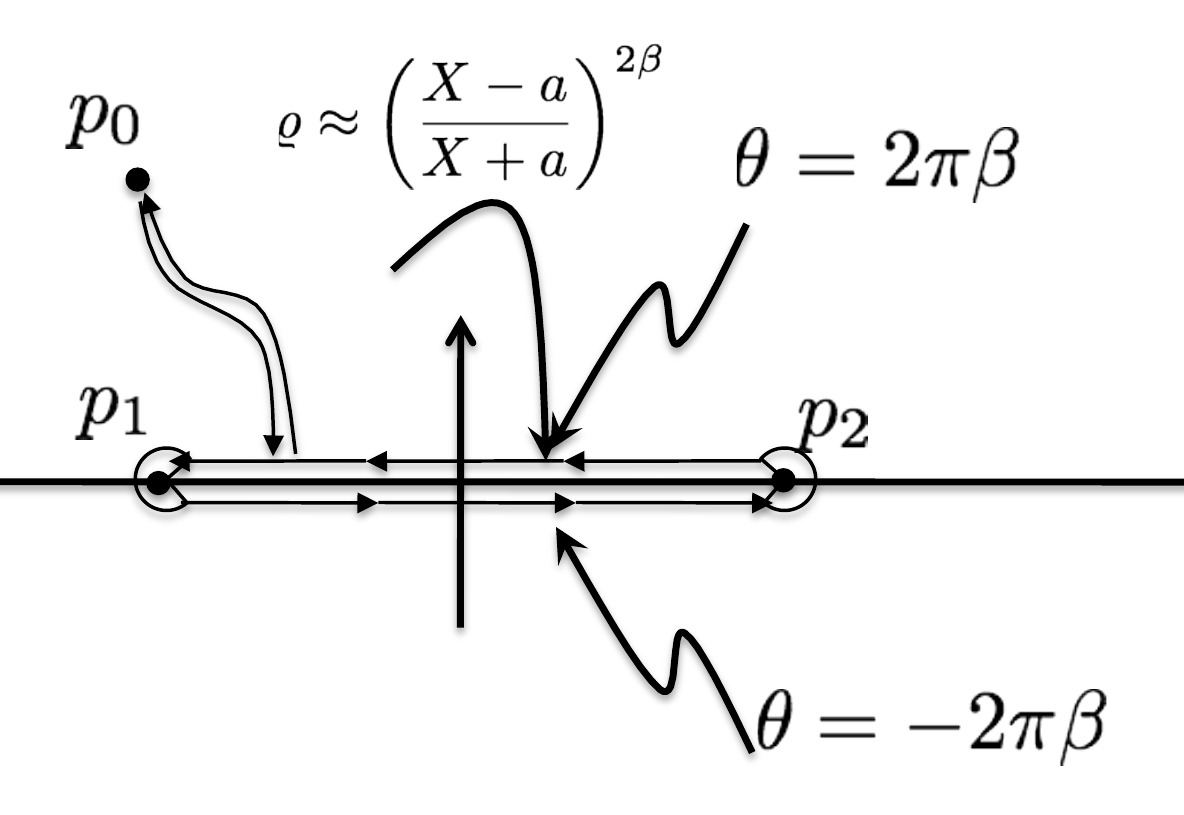}
\end{center}
\caption{The contour used in the calculation of the translational part of the monodromy of an edge dislocation.}
\label{fig:contour}
\end{figure}

As expected, $b_{g}$ is a parallel vector field, whose magnitude depends on the distance between the two disclinations and on the magnitude of the disclinations.


\subsection{Higher multipoles: the Stone-Wales defect}

Disclinations and edge dislocations are two-dimensional defects;
disclinations are generated by point sources (monopoles) of Gaussian curvature, whereas edge dislocations are generated by dipoles of Gaussian curvatures. Like in the electrostatic analog, the dipole moment is independent of the reference point if and only if the monopole vanishes. Unlike in the electrostatic analog, the dipole moment can be calculated exactly and not just asymptotically at infinity. This is quite interesting as the Liouville equation that relates the Gaussian curvature to the Laplacian of the conformal factor can be viewed as a nonlinear analog of the linear Gauss equation in electrostatics. Therefore, its solutions are expected to be less amenable to explicit computations.

In analogy to electrostatics, one is also interested in  defects that are higher order multipoles of Gaussian curvature. An example of such a defect is the Stone-Wales defect found in graphene, where four hexagonal cells transform into two pentagon-heptagon pairs \cite{SN88} (Figure~\ref{fig:stonewales}). Since pentagons and heptagons constitute disclinations of opposite signs, and since the configuration of the two  pentagon-heptagon pairs is anti-linear, the Stone-Wales defect constitutes a metric quadrupole.

\begin{figure}
\begin{center}
\includegraphics[height=1.8in]{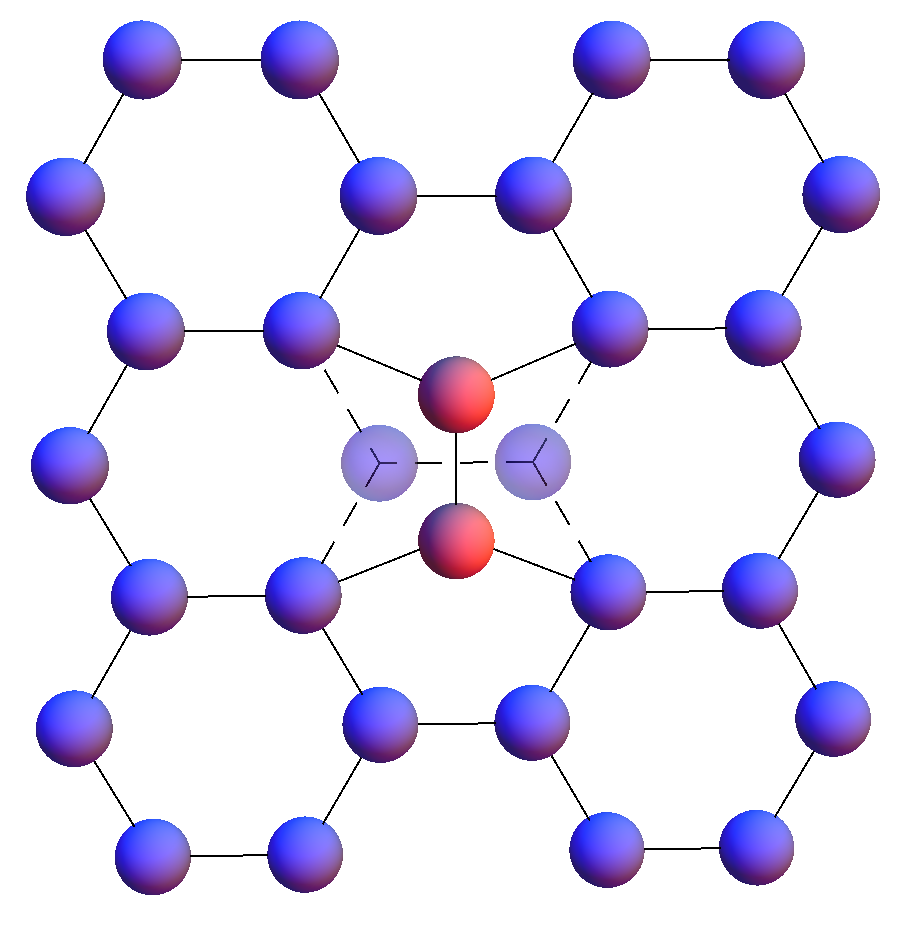}
\end{center}
\caption{The Stone-Wales defect: four hexagons are converted into two pentagon-heptagon pairs. }
\label{fig:stonewales}
\end{figure}

Suppose we removed a part of the lattice that contains the two pentagon-heptagon pairs. The remaining lattice would be perfect, i.e., isometrically embeddable in the plane. This is obvious as the Stone-Wales defect can be eliminated by a purely local change of lattice connectivities. This is not the case for a single pentagon-heptagon pair, where the defect can be detected at any distance from its locus. Thus, we expect a fundamental difference between metric monopoles and dipoles on the one hand, and higher metric multipoles on the other hand.

Motivated by the Stone-Wales defect, we consider a metric quadrupole, which can be realized using a two-dimensional conformal representation with four disclinations,
\beq
\conf(X,Y) = \beta \sum_{i=1}^4  s_i \,\log[(X-X_i)^2 + (Y-Y_i)^2],
\label{eq:SW}
\eeq
where $s_1 = -s_2 = -s_3 = s_4=1$, and $(X_1,Y_1)=(a,a)$, $(X_2,Y_2)=(-a,a)$, $(X_3,Y_3)=(a,-a)$, and $(X_4,Y_4)=(-a,-a)$. Correspondingly, the angle $\theta$ between the parallel frame and the parametric frame is
\[
\theta(X,Y) = 2\beta \sum_{i=1}^4  s_i \,\tan^{-1}\brk{\frac{Y-Y_i}{X-X_i}}.
\]
The branch cuts of the four addends can be chosen such that $\theta$ is smooth outside the rectangle whose vertices are the singular points,
see Figure~\ref{fig:stonewales2}.

\begin{figure}
\begin{center}
\includegraphics[height=2in]{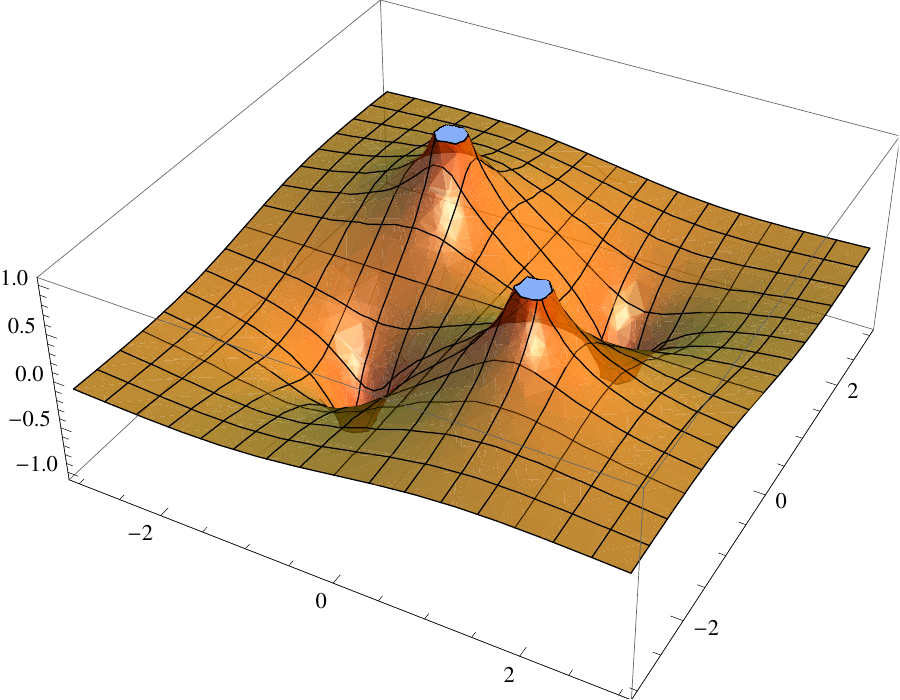}
\includegraphics[height=2in]{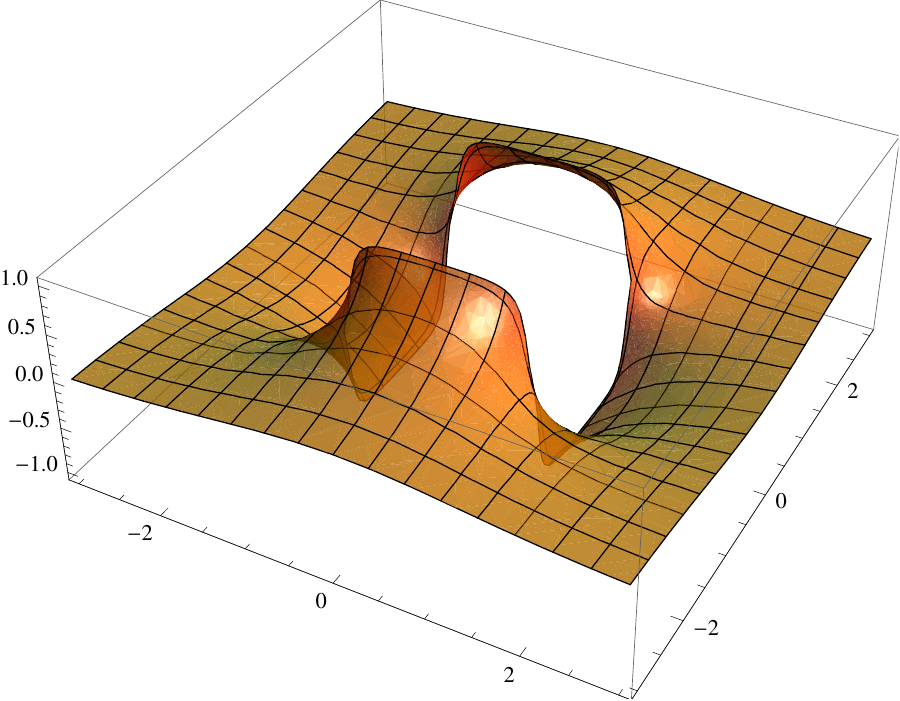}
\end{center}
\caption{$\conf(X,Y)$ and $\theta(X,Y)$ for a metric quadrupole with $\beta=1/4$ and $a=1$}
\label{fig:stonewales2}
\end{figure}

Like for a doubly-punctured plane, we will restrict our attention to loops that encircle all four singular points. Namely, we redefine the defect locus $\calD$ to be the rectangle $[-a,a]\times[-a,a].$ Then symmetry considerations show that  the monodromy is trivial. This implies, as we  show, that the defect is local in the following sense: there exists compact subset $K\subset\M$ such that $\M\setminus K$ embeds isometrically in the plane. In other words, except for a bounded region around the loci of the defect, the entire surface can be embedded isometrically (not only locally) in the Euclidean plane.

\begin{theorem}
Let $\M$ be $\bbR^2\setminus [-a,a]^2$ equipped with the Riemannian metric determined by the conformal factor \eqref{eq:SW}. Then there exists a compact set $K\subset\M$ such that $\M\setminus K$ embeds in the plane.
\end{theorem}

\begin{proof}
We will exploit the relation between conformal coordinates and  complex manifolds. Denote $Z = X + i\, Y$. Since the monodromy is trivial, the developing map descends to a map
\[
\dev : \M\to T_{p_0}\M,
\]
for any $p_0\in\M$. If
\[
\dev(X,Y) = x(X,Y)\,\partial_X|_{p_0} + y(X,Y)\,\partial_Y|_{p_0},
\]
then set $\tdev:\bbC\to\bbC$ to be
\[
\tdev(Z) = x(X,Y) + i\, y(X,Y).
\]
It is easy to see that
\beq
\tdev(Z) = e^{-\tconf(Z_0)} \int_{Z_0}^Z e^{\tconf(W)} \, dw,
\label{eq:complex_dev}
\eeq
where $\tconf:\bbC\to\bbC$ is the complexified conformal factor,
\[
\tconf(Z) = \varrho + i \theta = 2\beta \log \frac{(Z-Z_{1})(Z+Z_{1})}{(Z-Z_{2})(Z+Z_{2})}.
\]
The trivial monodromy implies that \eqref{eq:complex_dev} is path-independent.

The developing map is a local isometry. It remains to show that up to the possible exclusion of a compact set it is one-to-one, i.e., that $\tdev$ is one-to-one in a neighborhood of infinity.   Note that
\[
\tdev'(Z) = e^{\tconf(Z)-\tconf(Z_0)},
\]
which has a non-zero limit $e^{-\tconf(Z_0)}$ at infinity. Hence,
$\tdev(Z)$ has an expansion
\[
\tdev(Z) = e^{-\tconf(Z_0)}Z  + \sum_{n=0}^\infty \frac{\alpha_n}{Z^n}
\]
in a neighborhood of infinity.
Using the classical inversion,
\[
g(Z) = \frac{1}{\tdev(1/Z)},
\]
we have
\[
g'(Z) =  \frac{\tdev'(1/Z)}{Z^2\,\tdev^2(1/Z)}
\]
hence
\[
\lim_{Z\to0} g'(Z) = \lim_{Z\to\infty}  \frac{\tdev'(Z)}{Z^{-2}\,\tdev^2(Z)} = e^{\tconf(Z_0)}.
\]
It follows that $g'$ has a removable singularity at zero and a non-zero limit. By the inverse function theorem, $g$ is one-to-one in a neighborhood of zero, and therefore $\tdev$ is one-to-one in a neighborhood of infinity. Note that this analysis is applicable for every conformally represented metric with trivial monodromy and a conformal factor that vanishes at infinity.
\end{proof}

\subsection{Two-dimensional defects with trivial monodromy: the general case}

The result of the previous subsection whereby a locally Euclidean surface with trivial monodromy embeds (excluding a compact set) in the Euclidean plane can be generalized without requiring the existence of a global system of isothermal coordinates. In fact, it can be formulated as a theorem for affine embeddings.

\begin{theorem}
\label{th:SW}
Let $\M$ be a connected affine manifold with boundary. Suppose that
\begin{enumerate}
\item $\M$ is geodesically complete, i.e., geodesics extend indefinitely unless they hit the boundary.
\item $\pi_1(\M) = \bbZ$.
\item \label{it:circ} $\partial\M$ is homeomorphic to a circle.
\item $\M$ has trivial monodromy.
\item There exists a simple closed curve in $\M$ that is not null homotopic (i.e., not contractible) and has winding number one.
\end{enumerate}
Then, there exists a compact subset $K\subset\M$, such that $\M\setminus K$ embeds affinely in the Euclidean plane.
\end{theorem}

Recall that the winding number of a closed curve in the plane is the total number of times that its tangent rotates. The winding number of a curve in $\M$ is well-defined because tangent vectors can be translated unambiguously to a joint reference point. In fact, the winding number of a closed curve is given by the winding number of its image under the developing map.

\begin{proof}
Since the proof is somewhat long and technical, we will break it into short steps:

\renewcommand{\theenumi}{\roman{enumi}}
\renewcommand{\labelenumi}{(\theenumi)}
\begin{enumerate}
\itemsep2pt \parskip0pt \parsep0pt

\item $\M$ is an affine manifold. A locally Euclidean metric can be defined on $\M$ by  prescribing an inner-product on $T_p\M$ for some arbitrary point $p$, and parallel transporting tangent vectors to $p$; Assumption~4 guarantees that parallel transport is path-independent.

\item Geodesic completeness implies metric completeness. This is the well-known Hopf-Rinow theorem (\cite{DoC92}, pp. 146--149). Note  that the classical theorem is for a geodesically complete manifold without boundary. It is not hard to generalize the Hopf-Rinow theorem to manifolds with boundary: A manifold is metrically complete if and only if geodesics extend indefinitely unless they hit the boundary.

\item It follows from Assumption~4 that the developing map descends to a function $\dev :\M\to T_p\M\cong\R^2$, for a reference point $p\in\M$. Indeed, for $q\in\M$ let $\gamma$ be a curve in $\M$ connecting $p$ to $q$. Then
\[
\dev(q) = \int \Pi_{\gamma(t)}^p(\dot{\gamma}(t))\, dt.
\]
The triviality of monodromy implies that this integral does not depend on the chosen curve.

\item\label{it:la} By Proposition~\ref{prop:ddev} $\dev$ is a locally affine map  (also a local isometry), and in particular, a local diffeomorphism in $\Int(\M)$.

\item Even though $\partial\M$ is homeomorphic to a circle, its image under $\dev$ is not necessarily a simple curve. Let $L\subset\R^2$ be an open disc that contains $\dev(\partial\M)$, and let $K = \dev^{-1}(L)$.

\item $\overline{K}$ is bounded and connected by the following argument: By completeness and connectedness of $\M,$ and a version of Hopf-Rinow for manifolds with boundary, any point of $\M$ and, in particular, any point of $\overline K$ can be connected to $\partial M$ by a geodesic $\gamma.$ Since $\dev$ is locally affine by step~(\ref{it:la}), it follows that $\bar\gamma = \dev \circ \gamma$ is a geodesic. Since $\overline L$ is convex, and the endpoints of $\bar \gamma$ belong to $\overline L,$ we conclude that $\bar \gamma \subset \overline L.$ Thus the length of $\bar \gamma$ is less than the diameter of $L.$ It follows that the length of $\gamma$ is also less than the diameter of $L.$ So, the distance from any point of $\overline K$ to  $\partial \M$ is bounded. Moreover, $\partial \M$ is compact by assumption~\ref{it:circ} and thus bounded. Therefore, $\overline K$ is bounded as claimed. Similarly, the fact that $\bar \gamma \subset \overline{L}$ implies that $\gamma \subset \overline{K}.$ So any point in $\overline K$ can be connected within $\overline K$ to a point of $\partial \M.$ Moreover, $\partial \M$ is connected by assumption~\ref{it:circ}. So $\overline K$ is connected as claimed.

\item $\overline{K}$ is compact because it is closed and bounded.

\item\label{it:union} It follows from step~(\ref{it:la}) and the Implicit Function Theorem that $\overline K$ is a manifold with boundary and
\[
\partial \overline K = \partial\M \cup \dev^{-1}(\partial L).
\]
In particular, $\partial \overline K$ is a union of circles.

\item $\M\setminus K$ is complete because it is a closed subset of a complete manifold.

\item It follows from the proof of Lemma~3.3 in Chapter 7 of \cite{DoC92} that
\[
\dev: \M\setminus K \to\R^2\setminus L
\]
is a covering map; denote by $d$ the degree of the covering.

\item $\partial(\M\setminus K)$ is closed and bounded hence compact.

\item It follows from the classification of covering spaces together with the previous step that $\M\setminus K$ is a union of annuli.

\item The second homology group $H_2(\M)$ vanishes by the following argument: The interior of $\M$ is non-compact, so Proposition 3.29 from~\cite{Hat02} implies its second homology is trivial. But $\M$ is homotopy equivalent to its interior, so its homology is the same.

\item It follows from the previous step and a Mayer-Vietoris argument that $H_1(K) \simeq \bbZ$.

\item By classification of surfaces, $K$ is a (single) annulus.

\item By the previous step, $\partial K$ is a union of two circles. So, by step~(\ref{it:union}) we have $\partial(\M\setminus K) \simeq S^1$.

\item Hence $\M\setminus K$ is an annulus.

\item Since $\M$ is the union of two annuli, $\M\setminus K$ and $K$, along their joint boundary, it follows that $\M$ is an annulus.

\item Let $\gamma$ be a simple closed curve in $\R^2\setminus L$. It has winding number one.

\item Let  $\gamma^d$ be the concatenation of $\gamma$ with itself $d$ times; it has winding number $d$.

\item Let $\tg$ be a lifting of $\gamma^d$ to $\M\setminus K$, i.e., $\dev\circ\tg = \gamma^d$. Then by covering space theory, $\tg$ is a simple closed curve that generates $\pi_1(M) \simeq \bbZ$. Moreover, since lifting preserves winding number, the winding number of $\tg$ is $d$.

\item By classification of simple closed curves on an annulus (compare with the case of the torus in Chapter~1 of~\cite{FM11}), every two non-contractible simple closed curves on a annulus are homotopic. By Proposition 1.10 of~\cite{FM11}, every two simple closed curves in a surface that are homotopic are also isotopic. Together with Assumption~5 it follows that $\tg$ is isotopic to a curve that has winding number one, hence $d=1$.

\item Since $\dev$ is a covering of degree one it is bijective. A bijective map that is locally affine is an affine embedding; this concludes the proof.

\end{enumerate}

\end{proof}

\subsection{Point defects}

We next briefly consider point defects, i.e., manifolds in which the locus of the defect is a point, or a finite collection of points. In three dimensions, $\M$ is simply connected, which implies a trivial fundamental group, and therefore a trivial monodromy. The implication is that point defects cannot be detected in the same way as line defects, by making metric measurements around loops that encircle the defect.

Consider, for example, a point defect of type vacancy. The Volterra cut-and-weld procedure in this case would be to remove from $\R^3$ a ball of radius $a$ (say, centered at the origin) and weld the boundary of this ball into a single point. $\M$.
Here we take $\M=\R^3\setminus\{0\}$, which we parametrize using spherical coordinates,
\[
(R,\Theta,\Phi) \in (0,\infty)\times[0,\pi]\times[0,2\pi),
\]
and the metric is
\[
\g = dR\otimes dR + (R+a)^2\,d\Theta\otimes d\Theta + (R+a)^2\,\sin^2\Theta\, d\Phi\otimes d\Phi,
\]
The vacancy manifests in that the intrinsic geometry in its vicinity has ``too much length"; the surface area of spheres that converge to a single point---the locus of the defect---does not tend to zero.

Here too, like for the two-dimensional defects considered in the last subsection,  the trivial monodromy manifests itself in the fact that if a compact set that contains the locus of the defect is removed, the punctured manifold embeds isometrically in Euclidean space.

\section{Discussion}
\label{sec:discussion}

This paper is concerned with the description of singular defects in isotropic media. The geometric structure of such materials is fully encoded in their reference metric.
We showed that topological defects can be fully described by geometric fields that  only reflect  the metric structure. More specifically, topological defects are described by the affine structure induced by the locally Euclidean metric. In particular, dislocations are described by the affine structure without reference to torsion.

It is important to stress that our approach does not contradict former approaches.
As showed by Wang almost 50 years ago, non-Riemannian material connection arise naturally in media that exhibit discrete symmetries. Our statement is that the
``failure of parallelograms to be closed" can be fully captured in isotropic media by the translational part of the monodromy, which only depends on the postulated affine structure.
Note that
the moment that topological defects are encoded by a metric, they can be realized without the need to break any structure. For example, both disclinations and edge dislocations can be created by imposing a two-dimensional reference metric via differential swelling, as in \cite{KES07}, with a  swelling factor that is harmonic everywhere but at a finite number of points.

Another interesting observation is the relation between disclinations and dislocations as monopoles and dipoles of curvature charges, and their electrostatic analog. Like in the electrostatic analog, the dipole moment, is independent of the reference point if and only if the monopole moment vanishes, thus answering the question raised by Miri and Rivier \cite{MR02} about the coexistence of disclinations and dislocations. Yet, unlike in the electrostatic analog, the dipole moment (reflected by the translational part of the monodromy) can be calculated exactly, and not just asymptotically at infinity, by metric measurements around loops that encircle the locus of the defect.  In this respect, the nonlinear Laplace equation that connects  the metric to the Gaussian curvature turns out to be  ``simpler" than the linear Poisson equation in electrostatics.

Another  result of the present paper is that every two-dimensional metric that has trivial monodromy can be isometrically embedded in the plane, up to the possible need to exclude a compact subset. In  practical terms, this means that every defective plane in which the defect cannot be detected by metric measurements along a curve that encircles the locus of the defect can be embedded in Euclidean plane with metric distortions restricted to a compact set. This observation, which is relevant for example to  known defects in graphene, has immediate implications on the elastic energy associated with such defects.

Our paper is concerned with the description of defects, and does not present calculations of residual stresses in bodies with defects. To calculate stresses one needs a concrete model, e.g., a neo-Hookean solid as used in \cite{YG12}, with strain measured relative to the postulated reference metric.

\paragraph{Acknowledgements}
We are grateful to Leonid Polterovich and Eran Sharon for  fruitful discussions.

\bibliographystyle{spmpsci}
\bibliography{MyBibs}

\end{document}